\documentclass[reqno,12pt]{amsart}
\newtheorem{theorem}{Theorem}[section]
\newtheorem{theorem*}[theorem]{Theorem}
\newtheorem{lemma}[theorem]{Lemma}
\newtheorem{definition}[theorem]{Definition}

\newtheorem{proposition}[theorem]{Proposition}

\theoremstyle{remark}
\newtheorem{remark}[theorem]{Remark}
\numberwithin{equation}{section}

\usepackage{array}

\newcolumntype{P}[1]{>{\centering\arraybackslash}p{#1}}
\usepackage{makecell}
\usepackage{float}

\usepackage{multirow}
\usepackage{array,booktabs}
\usepackage{todonotes}
\usepackage[margin=1in]{geometry}
\usepackage{verbatim}
\usepackage{amsmath}
\usepackage{amsthm}
\usepackage{amsfonts}
\usepackage{amssymb}
\usepackage{amsthm}
\usepackage{enumerate}
\usepackage{amsrefs}
\usepackage{enumerate}
\usepackage{amsrefs}
\usepackage{tensor}
\usepackage{epic}
\usepackage{mathtools}
\usepackage{tikz}
\usetikzlibrary{chains}
\usepackage{graphicx}
\usepackage[font=small,skip=0pt]{caption}
\usepackage{makecell}
\usepackage{tensor}
\usepackage[super]{nth}
\usepackage{tabulary}
\usepackage{lipsum}
\usepackage{cases}
\usepackage{accents}

\usepackage{dynkin-diagrams}

\title{Commutativity equations and their trigonometric solutions}
\author{Maali Alkadhem and Misha Feigin }

\begin{document}

\begin{abstract}
We consider commutativity equations $F_i F_j =F_j F_i$ for a function $F(x^1, \dots, x^N),$ where $F_i$ is a matrix of the third order derivatives $F_{ikl}$. 
We show that under certain non-degeneracy conditions a solution $F$ satisfies the WDVV equations. Equivalently, the corresponding family of Frobenius algebras has the identity field $e$.

We study trigonometric solutions $F$
 determined by a finite collection of vectors with multiplicities,
 and we give an explicit formula for $e$ for all the known such  solutions. The corresponding collections of vectors are given by non-simply laced root systems or are related to their projections to the intersection of mirrors.
\end{abstract}
\maketitle

\section{Introduction}
 
A celebrated system of the Witten--Dijkgraaf--Verlinde--Verlinde (WDVV) equations for a prepotential function $F(x)=F(x^1,\dots, x^N)$ has the form
\begin{equation}\label{WDVV}
    F_{ijk}g^{kl}F_{lmn}=F_{mik}g^{kl}F_{ljn},
\end{equation}
where
\begin{align*}
F_{ijk} = \frac{\partial^{3} F}{\partial x^i \partial x^j \partial x^k},
\end{align*}
and $G=(g^{kl})$ is a constant symmetric $N\times N$ matrix.
These equations appeared in topological field theories \cites{DVV.1990, Witten.1990} and they are in the core of Frobenius manifolds theory \cite{Dubrovin.1996}.  
%
In these considerations one normally has the property that the components of the flat metric $G^{-1}$ can be represented as
\begin{equation}\label{G and e}
(G^{-1})_{ij}=\sum_{k=1}^{N}e^k F_{kij}
\end{equation}
for some vector field $e=\sum_{k=1}^{N}e^k(x) \partial_{x^k}$
which is the identity field for the corresponding family of Frobenius algebras. %
For example, in the Frobenius manifolds theory one normally has $e=\partial_{x^1}$, which is flat with respect to the metric  $G^{-1}$. In the case of almost dual Frobenius manifold on the space of orbits of a finite Coxeter group the field $e$ is not constant, it is proportional to the Euler vector field \cite{Dubrovin 2004}.

It is also of interest to consider equations \eqref{WDVV} without the additional assumption \eqref{G and e} which expresses the metric $G^{-1}$ as a linear combination of the third order derivatives of the prepotential.
Indeed, in the case of $G$ being the identity matrix the corresponding equations \eqref{WDVV} have the form of the commutativity equations 
\begin{equation}\label{commutativityEq}
    F_i F_j =F_j F_i,
\end{equation}
where $F_i$ is the $N\times N$ matrix with matrix entries $(F_i)_{kl}=F_{ikl} = \frac{\partial^{3} F}{\partial x^i \partial x^k \partial x^l}.$

Equations \eqref{commutativityEq} appeared in the study of $\mathcal{N}=4$ supersymmetric mechanical system (see \cite{Wyllard 2000}).
For a suitable ansatz for the supercharges the supersymmetry algebra relations are satisfied provided that equations \eqref{commutativityEq} hold.
Existence of the identity field or rather, more specifically, additional relations of the form $\sum_i x^i F_{ijk}=-\delta_{jk}$ lead to further superconformal symmetry (see \cite{Wyllard 2000} and also \cite{Bellucci 2005}, where the relation with the WDVV equations is emphasized).

In this paper we are interested in commutativity equations \eqref{commutativityEq}
and the additional condition of the existence of a vector field $e=\sum_{i=1}^N e^k \partial_{x^k}$ such that $e(F_{ij})=\delta_{ij}$.
This vector field is the identity vector field for a family of algebras depending on $x$. One of our main results provides a sufficient condition on $F$ which ensures that $e$ exists.
The components of the field can then be expressed via determinants of the matrices whose entries are the third order derivatives of the prepotential $F$.
Similarly, for a general constant matrix $G$ we establish a  representation of $G^{-1}$ as a linear combination of the matrices $F_i$ as a \textit{consequence} of equations \eqref{WDVV} (see Sections \ref{CommEqAndWDVV} and \ref{Existence of e}).

There is an interesting class of solutions of the equations \eqref{WDVV}, \eqref{G and e} determined by finite collections $\mathcal{A}$ of vectors. The corresponding prepotential has the form 
\begin{equation}\label{F.log}
F=\sum_{\alpha\in\mathcal{A}}(\alpha,x)^{2}\log (\alpha,x),\quad x\in V.
\end{equation}
In the case when $\mathcal{A}$ is a root system such solutions of the WDVV equations appeared in \cite{Martini+Gragert 1999}. 
They are almost dual prepotentials for the finite group orbit spaces Frobenius manifolds \cite{Dubrovin 2004}.
Such solutions also appear in four-dimensional
Seiberg--Witten theory as perturbative parts of the corresponding prepotentials \cite{MMM.2000}.
More generally, solutions of the form  \eqref{F.log} exist for special configurations of vectors known as $\vee$-systems introduced by Veselov in \cite{Veselov 1999}.
This class of of solutions was studied further in \cites{Chalykh+ Veselov 2001, Misha&Veselov 2007, Misha&Veselov 2008, Schreiber+ Veselov 2014}.
Thus it was shown that the class is closed under the operations of taking subsystems and projections of $\mathcal{A}$, and such solutions have to do with Dubrovin's almost duality on the discriminant strata.
Connection of these solutions to the supersymmetric mechanics was explored in \cite{GLK 2009}.
More generally, one may also consider solutions of the form \eqref{F.log} for the commutativity equations \eqref{WDVV} (without extra condition \eqref{G and e}).
The corresponding (irreducible) configurations of vectors $\mathcal{A}$ can be shown to be the complex Euclidean version of $\vee$-systems introduced in \cite{Misha&Veselov 2008}.

There are also interesting trigonometric solutions of the equations \eqref{WDVV}, \eqref{G and e} of the form 
\begin{equation}\label{F.trig}
    F=\sum_{\alpha\in\mathcal{A}}c_{\alpha}f((\alpha,x))+Q(x,y),
\end{equation}
where function $f=f(z)$ satisfies $f^{'''}(z)=\cot z,\, c_{\alpha} \in \mathbb{C}$ and $Q$ is a cubic polynomial depending on the additional variable $y$.
Solutions of this form for reduced root systems and Weyl-invariant multiplicities were obtained by Hoevenaars and Martini in \cite{Martini 2003} (see also \cite{Shen 2019} and \cite{George+Misha 2019} for more details). %
They appear as almost dual prepotentials for the extended affine Weyl groups orbit spaces \cites{ Dubrovin+ Zhang 1998,  Dubrovin+Strachan+ Zhang+Zuo 2019}, see \cite{Riley+ Strachan 2007} for type $A_N$. 
Such solutions also appeared in five-dimensional Seiberg--Witten theory as perturbative parts of prepotentials \cite{MMM.2000}. 
In the case of simply laced root systems these solutions describe quantum cohomology of resolutions of simple $A,D,E$ singularities \cite{Bryan 2008}.
Solutions of the form \eqref{F.trig} for general configurations $\mathcal{A}$ were initially studied in \cite{Misha2009} where a closely related notion of the trigonometric $\vee$-system was introduced. 
Similarly to the rational case, we showed in \cite{Maali+Misha 2021} 
that this class of solutions is closed under restrictions and that a subsystem of a trigonometric $\vee$-system is also a trigonometric $\vee$-system.
The restriction procedure for the classical root systems 
recovers solutions obtained by Pavlov from reductions of Egorov hydrodynamic chains \cite{Pavlov 2006}.

There are also elliptic versions of some of these solutions considered by Riley and Strachan in \cites{Riley+Strachan 2006, Strachan 2010}.

It appears that solutions of the form \eqref{F.trig} with $Q=0$ of the WDVV equations \eqref{WDVV}, \eqref{commutativityEq} may also exist.
Such a solution for the root system $B_N$ appeared in \cite{Martini 2003} and it was generalized to $BC_N$ in \cite{MGM 2020}.
The corresponding metric $G$ is the identity so the commutativity equations \eqref{commutativityEq} hold as well.

Solutions of the form \eqref{F.trig}
with $Q=0$ for the commutativity equations \eqref{commutativityEq} for the root systems $\mathcal{A}=F_4, G_2$
were obtained in \cite{Maali+Misha 2021}. 
The corresponding multiplicities are Weyl invariant but they have to satisfy a linear relation.
A multi-parameter deformation of the solution for the root system $BC_N$ was also obtained in  \cite{MGM 2020}.
It is unclear whether there are more Frobenius manifold structures associated with such solutions.

In this paper we study solutions of the commutativity equations \eqref{commutativityEq}
of the form \eqref{F.trig} with $Q=0$.
Thus we give a $\vee$-system version of conditions which the corresponding configuration of vectors has  to satisfy, which we call {\it a Euclidean trigonometric $\vee$-system} (see Section \ref{section.general solution of FF}).
We also show that restrictions of solutions of the commutativity equations give new solutions and that a subsystem of a Euclidean trigonometric $\vee$-system is also a Euclidean trigonometric $\vee$-system (see Sections \ref{section.subsystem of Euclid.trig}, \ref{section.restrictions of Euclid.trig}).
In Section \ref{Section.Relations} we clarify relations of Euclidean trigonometric $\vee$-systems to other versions of rational and trigonometric $\vee$-systems.
All the known irreducible solutions of the commutativity equations \eqref{commutativityEq} of the form \eqref{F.trig} with $Q=0$ are the non-simply laced root systems $BC_N, F_4, G_2$ with a relation between invariant multiplicities as well as restrictions of such solutions to the intersection of mirrors (in the case of $BC_N$ one can also extend analytically integer parameters defining the restriction).
In all these cases we give an explicit uniform formula for the corresponding identity field $e$ in Section \ref{Section.Uniform formula}.
Existence of the identity field implies that we also get new solutions of WDVV equations \eqref{WDVV}, \eqref{commutativityEq} in the case of root system $F_4$ and its projections.


\section{Commutativity equations and Euclidean trigonometric $\vee$-systems}\label{section.general solution of FF}
Let $\mathcal{A}$ be a finite set of non-zero vectors in a Euclidean space $V\cong\mathbb{C}^N$, $N\in \mathbb N$,  with the bilinear inner product $(\cdot , \cdot)$.
Let $c\colon \mathcal{A}\to \mathbb{C}$ be the (multiplicity) function. We denote 
$c_{\alpha}\coloneqq c(\alpha)$ for  $\alpha\in \mathcal{A}$.
We assume that $\mathcal{A}$ belongs to a lattice of rank $N$.
For each vector $\alpha \in \mathcal{A}$ let us introduce the set of its collinear vectors from $\mathcal{A}$:
%
$$    \delta_{\alpha} \coloneqq \{\gamma  \in \mathcal{A} \colon \gamma 	\sim \alpha   \}.$$
%
Let $\delta \subseteq \delta_{\alpha}$ and $\alpha_{0}\in \delta_{\alpha}.$
Then for any $\gamma \in \delta$ we have $\gamma= k_{\gamma}\alpha_{0}$ for some $k_{\gamma}\in \mathbb{R}.$ 
Note that $k_{\gamma}$ depends on the choice of $\alpha_{0}$ and different choices of $\alpha_{0}$ give rescaled collections of these parameters.
Define $C_{\delta}^{\alpha_{0}}\coloneqq \displaystyle \sum_{\gamma\in \delta}c_{\gamma}k_{\gamma}^{2}.$ 
Note that  $C_{\delta}^{\alpha_{0}}\neq 0$ if and only if  $C_{\delta}^{\widetilde{\alpha}_{0}}\neq 0$ for any $\widetilde{\alpha}_{0} \in \delta.$
We define {\it strings} (or {\it series}) of vectors as follows (cf.  \cite{Misha2009}).
%

For any $\alpha \in \mathcal{A}$ let us distribute all the vectors in $\mathcal{A}\setminus \delta_{\alpha}$ into a disjoint union of $\alpha$-strings
$$\mathcal{A}\setminus \delta_{\alpha}= \bigsqcup_{s=1}^{k}  \Gamma_{\alpha}^{s},$$
where $k \in \mathbb{N}$ depends on $\alpha.$ These stings $\Gamma_{\alpha}^{s}$ are determined by the property that for any $s=1,\dots, k$ and for any two covectors $\gamma_{1},\gamma_{2}\in \Gamma_{\alpha}^{s}$ one has either $\gamma_{1}+\gamma_{2}=m\alpha$ or $\gamma_{1}-\gamma_{2}=m\alpha$ for some $m \in \mathbb{Z}.$
We assume that the strings are maximal, that is if $\gamma \in \Gamma_{\alpha}^{s}$ for some $s\in \mathbb{N},$ then $\Gamma_{\alpha}^{s}$ must contain all the covectors of the form $\pm \gamma+m\alpha\in \mathcal{A}$ with $m\in \mathbb{Z}.$ 
Note that if for some $\beta \in \mathcal{A}$ there is no $\gamma \in \mathcal{A}$ such that $\beta \pm \gamma=m\alpha$ for $m\in \mathbb{Z},$ then $\beta$ itself forms a single $\alpha$-string.

By replacing some vectors from $\mathcal{A}$ with their opposite ones and keeping the multiplicity unchanged one can get a new configuration whose vectors belong to a half-space. We will denote such a system by $\mathcal{A}_{+}.$ %
If this system contains repeated vectors $\alpha$ with multiplicities $c_{\alpha}^{i}$ then we replace them with the single vector $\alpha$ with multiplicity $c_{\alpha}\coloneqq \sum_{i} c_{\alpha}^{i}.$
%
%
%

Let us now define \textit{a Euclidean trigonometric $\vee$-system} in analogy with a trigonometric $\vee$-system \cites{Misha2009}.
\begin{definition}\label{Euclidean trig.V-sys defintion}
The pair $(\mathcal{A},c)$ is called a Euclidean trigonometric $\vee$-system if for all $\alpha \in \mathcal{A}$ and for all $\alpha$-strings $\Gamma_{\alpha}^{s},$ one has the relation
\begin{equation}\label{Euclid.Trig.condition}
\sum_{\beta\in \Gamma_{\alpha}^{s} }c_{\beta}(\alpha,\beta)\alpha \wedge \beta=0.
\end{equation}
\end{definition}
%
%
%

%
Consider a function $F$ given by the formula
\begin{equation}\label{trig.solution. general form 2}
    F=\sum_{\alpha \in \mathcal{A}}c_{\alpha} f(( \alpha,x)),
\end{equation}
where 
the function $f$ 
is given by
\begin{equation*}
    f(z)= \frac{1}{6} i z^3+\frac{1}{4} Li_{3}(e^{-2iz})
\end{equation*}
so that $f^{\prime \prime \prime}(z)=\cot z$.
We are interested in configurations $(\mathcal{A},c)$ 
such that the commutativity equations
 \begin{equation}\label{FF=FF.1}
  F_i F_j = F_j F_i, \quad i,j =1,\dots,N,  
 \end{equation}
hold,
where $F_i$ is the $N\times N$ matrix with entries
$$(F_{i})_{pq}=F_{ipq} = \frac{\partial^{3} F}{\partial x^i \partial x^p \partial x^q}.$$
%
The following statement establishes invariance of the commutativity equations under the action of the group of orthogonal transformations $O(N, \mathbb{C})$.
Summation from 1 to $N$ over repeated indices  will be assumed throughout unless indicated otherwise.
\begin{proposition}\label{FF under orthogonal transformation}
Suppose a function $F=F(x^1 ,\dots , x^N)$ satisfies commutativity equations \eqref{commutativityEq}.
Let $C=(C_i^k)\in O(N, \mathbb{C}) $, and let
\begin{align}\label{changing variables.V2}
\widetilde{x}^k ={C} _{i}^{k} x^i ,
\end{align}
where $\widetilde{x}^1 ,\dots ,\widetilde{x}^N$ is a new coordinates system.
Then $\widetilde{F}(\widetilde{x})=F(x)$ satisfies commutativity equations
\begin{align}\label{comm.in new coord}
   \widetilde{F}_i   \widetilde{F}_j = \widetilde{F}_j  \widetilde{F}_i, \quad i,j=1,\dots,N, 
\end{align}
where $(\widetilde{F}_i)_{pq} = \frac{\partial^3 \widetilde{F}}{\partial \tilde x^i \partial \tilde x^p\partial \tilde x^q }$.
\end{proposition}
\begin{proof}
Since 
$\partial_{x^i}=C_i^k \partial_{\widetilde{x}^k}$,
 we have
$
F_{ijk}={C}_{k}^{\widetilde{k}}{C}_{j}^{\widetilde{j}}{C}_{i}^{\widetilde{i}}\widetilde{F}_{\widetilde{i}\widetilde{j}\widetilde{k}}.   
$
Then commutativity equations
$
F_{ijk} F_{klm}=F_{mjk}F_{kli}
$
in the new coordinates take the form
\begin{align}\label{comm.2}
    {C}_{k}^{\widetilde{k}}{C}_{j}^{\widetilde{j}}{C}_{i}^{\widetilde{i}} {C}_{k}^{a}{C}_{l}^{b}{C}_{m}^{d}
    \widetilde{F}_{\widetilde{i}\widetilde{j}\widetilde{k}}   \widetilde{F}_{abd}
    = {C}_{m}^{\widetilde{m}}{C}_{j}^{\widetilde{j}}{C}_{k}^{\widetilde{k}} {C}_{k}^{a}{C}_{l}^{b}{C}_{i}^{d}
    \widetilde{F}_{\widetilde{m}\widetilde{j}\widetilde{k}}   \widetilde{F}_{abd}.
\end{align}
Now we multiply both sides of equality \eqref{comm.2} by $\widehat{C}_{\alpha}^{m}\widehat{C}_{\beta}^{j}\widehat{C}_{\gamma}^{l} \widehat{C}_{\epsilon}^{i}  $, where
$\widehat{C}= C^{-1}$ so that $\widehat{C}_{\alpha}^{k} C_{k}^{\beta} =\delta_{\alpha}^{\beta}.$
We get
\begin{align}\label{comm.3}
C_{k}^{\widetilde{k}} C_{k}^{a} \widetilde{F}_{\epsilon \beta \widetilde{k}} \widetilde{F}_{a \alpha \gamma}     
    =   C_{k}^{\widetilde{k}} C_{k}^{a} \widetilde{F}_{\alpha \beta \widetilde{k}} \widetilde{F}_{a \gamma \epsilon} .
\end{align}
For an orthogonal transformation $C$ we have $C_{k}^{\widetilde{k}} C_{k}^{a}= \delta^{ \widetilde{k} a}$. Hence equality \eqref{comm.3} reduces to \eqref{comm.in new coord}.
 \end{proof}
We are going to establish a relation between
solutions \eqref{trig.solution. general form 2} of the commutativity equations \eqref{FF=FF.1} and Euclidean trigonometric $\vee$-systems. The following two lemmas hold. 
\begin{lemma}\label{main 1.FF}
The commutativity equations \eqref{FF=FF.1} for the function \eqref{trig.solution. general form 2} are equivalent to the identity
\begin{equation}\label{identity 1 for FF}
    \sum_{\alpha,\beta \in \mathcal{A}}c_{\alpha}c_{\beta}(\alpha,\beta)\cot(\alpha,x)\cot(\beta,x)B_{\alpha,\beta}(a,b) \alpha \wedge \beta =0,
\end{equation}
for all $a,b \in V,$ where $B_{\alpha,\beta}(a,b)=\alpha\wedge\beta=(\alpha,a)(\beta,b)-(\alpha,b)(\beta,a).$
\end{lemma}
\begin{lemma} \label{iden 2.FF.with collinear vectors} 
%
Suppose that identity \eqref{identity 1 for FF} holds
for any $a,b \in V$. Suppose also that
$C_{\delta}^{\alpha_{0}}\neq 0$ for any $\alpha \in \mathcal{A},\, \delta\subseteq \delta_{\alpha}, \,  \alpha_{0}\in \delta_{\alpha}.$ 
%
%
Then $\mathcal{A}$ is a Euclidean trigonometric $\vee$-system.
\end{lemma}
%
Proofs of Lemmas \ref{main 1.FF} and \ref{iden 2.FF.with collinear vectors} are similar to the proofs of analogous statements  in \cite{Maali+Misha 2021} for the case of the trigonometric $\vee$-system (see also \cite{Maali.2022}).

Note that if $\mathcal{A}$ is a Euclidean trigonometric $\vee$-system then the left-hand side of identity \eqref{identity 1 for FF} is non-singular. 
Since all vectors from $\mathcal{A}$ belong to an $N$-dimensional lattice then the left-hand side of identity \eqref{identity 1 for FF} is a rational function in suitable exponential variables which has degree zero and therefore is a constant.
In order to find this constant, by changing some of the vectors from $\mathcal{A}$ to their opposite ones we can assume that all vectors from $\mathcal{A}$ belong to a half-space, hence form a positive system $\mathcal{A}_{+}$.
Then in an appropriate limit in a cone $\cot(\alpha,x)\to i$ for all $\alpha\in \mathcal{A}_{+}$
and the identity  \eqref{identity 1 for FF} reduces to 
$$\sum_{\alpha,\beta\in \mathcal{A}_{+}}c_{\alpha}c_{\beta
}(\alpha,\beta) B_{\alpha,\beta}(a,b) \alpha \wedge \beta =0.$$
From these considerations 
we get the following result. 
\begin{theorem}\label{FF and Euclidean trig}
Suppose that a configuration $(\mathcal{A},c)$ satisfies the condition 
$C_{\delta}^{\alpha_{0}}\neq 0$ for any $\alpha \in \mathcal{A},\, \delta\subseteq \delta_{\alpha}, \,  \alpha_{0}\in \delta_{\alpha}.$ 
Then the commutativity equations \eqref{FF=FF.1} for the prepotential \eqref{trig.solution. general form 2}
imply the following two conditions:

(1) $\mathcal{A}$ is a Euclidean trigonometric $\vee$-system,

(2) $\sum_{\alpha,\beta\in \mathcal{A}_{+}}c_{\alpha}c_{\beta
}(\alpha,\beta) B_{\alpha,\beta}(a,b) \alpha \wedge \beta =0$ for all $a, b \in V$. 

Conversely, if a configuration $(\mathcal{A},c)$ satisfies conditions (1), (2) then commutativity equations \eqref{FF=FF.1} hold.
\end{theorem}
%
%
Root systems of Weyl groups provide examples of Euclidean trigonometric $\vee$-systems. 
\begin{proposition}\label{root system is Euclidean}
A root system $\mathcal{A}=\mathcal{R}$ with Weyl-invariant multiplicity function $c$ is a Euclidean trigonometric $\vee$-system.
\end{proposition}
\begin{proof}
Fix $\alpha \in \mathcal{R}.$ Take any $\beta \in \mathcal{R},$ and let $\gamma=s_{\alpha}\beta=\beta-\frac{2(\alpha,\beta)}{(\alpha,\alpha)}\alpha.$ 
%
%
Since $\frac{2(\alpha,\beta)}{(\alpha,\alpha)}\in\mathbb{Z}$ we get that
 $\beta, \gamma \in \Gamma_{\alpha}^{s}$ for some $s.$
%
%
We have 
$$c_{\beta}=c_{\gamma},\quad (\alpha, \beta)=-(\alpha,\gamma), \quad \alpha \wedge \beta = \alpha \wedge \gamma.$$
Hence the contribution of $\beta$ and $ \gamma$ to the sum in \eqref{Euclid.Trig.condition} cancel each other.
 
\end{proof}

In general root systems $\mathcal{A}=\mathcal{R}$ with invariant multiplicities do not satisfy condition $(2)$ in Theorem \ref{FF and Euclidean trig}.
%
It has been shown in \cite{George+Misha 2019} that this condition is satisfied for root systems $\mathcal{R}= BC_N, \, F_4 , \, G_2$ with special invariant multiplicities.
%
%

Solutions of commutativity equations can be applied to construct $\mathcal{N}=4$ supersymmetric mechanical systems.
%
Hamiltonians corresponding to root systems $\mathcal{R}= BC_N, \, F_4 , \, G_2$ were given explicitly in \cite{George+Misha 2019}.
%

\section{Subsytems of a Euclidean trigonometric $\vee$-system}\label{section.subsystem of Euclid.trig}

Now we consider subsystems of a Euclidean trigonometric  $\vee$-system. 
%
%
\begin{definition}\label{Subsys. of Euclidean trig.sys}
Let $\mathcal{A}\subset V$ be a finite collection of vectors.
A subset $\mathcal{B}\subseteq\mathcal{A}$ is called a subsystem if $$\mathcal{B}=\mathcal{A} \cap W$$ for a linear subspace $W\subseteq V.$ %
The subsystem $\mathcal{B}$ is called reducible if $\mathcal{B}$ is a disjoint union of two non-empty subsystems $\mathcal{B}=\mathcal{B}_{1} \sqcup\mathcal{B}_{2}$. 
The subsystem $\mathcal{B}$ is called \textit{irreducible} if it is not reducible.
\end{definition}
If $c$ is a multiplicity function for $\mathcal{A}$ then we will equip a  subsystem $\mathcal{B}\subseteq \mathcal{A}$ with the multiplicity function which is the restriction of the multiplicity function $c$ on $\mathcal{B}$.

Assume that the linear span $\langle \mathcal{B \rangle}=W.$
We say that the subsystem $\mathcal{B}$ is \textit{non-isotropic} if the restriction of the inner product $(\cdot, \cdot)$ onto $W$ is non-degenerate. 
%
 
%
\begin{theorem}\label{isotropic Euclid v-sys. is v-sys.}
Any non-isotropic subsystem of a Euclidean trigonometric $\vee$-system is also a Euclidean trigonometric $\vee$-system.
\end{theorem}
The proof of Theorem \ref{isotropic Euclid v-sys. is v-sys.} is similar to the proof in \cite{Maali+Misha 2021} of the analogous statement for the trigonometric $\vee$-system, see also \cite{Maali.2022}.
 
\section{Relation with other types of $\vee$-systems}\label{Section.Relations}
\subsection{Relation with trigonometric $\vee$-systems}\label{section.Euclid+Complex}
For a finite subset $\mathcal{A}\subset V$ with a multiplicity function $c\colon \mathcal{A} \to \mathbb{C},$ 
consider a bilinear form $G_{\mathcal{A},c}$ on $V$ given by
\begin{equation}\label{GA.for vectors}
G_{\mathcal{A},c}(x,y)=\sum_{\alpha\in \mathcal{A}}c_{\alpha}(\alpha, x)(\alpha , y),\quad x,y \in V,
\end{equation}
where $c_{\alpha}\coloneqq c(\alpha)$. 
%
Following an analogy with the rational case (see \cite{Misha&Veselov 2008} and subsection \ref{section.Euclid+Complex.1} below), we say that the pair $(\mathcal{A},c)$ is \textit{well-distributed} in $V$ if the
bilinear form \eqref{GA.for vectors} is proportional to the form $(\cdot,\cdot).$
The pair $(\mathcal{A},c)$ is called \textit{a trigonometric $\vee$-system} if it satisfies the relations 
\begin{equation}\label{Trig.condition.Vector}
\sum_{\beta\in \Gamma_{\alpha}^{s} }c_{\beta}G_{\mathcal{A},c}(\alpha, \beta)\alpha \wedge \beta=0
\end{equation}
for all $\alpha \in \mathcal{A}$ and all $\alpha$-strings $\Gamma_{\alpha}^{s}$.

Now let $(\mathcal{A},c)$ be a Euclidean trigonometric $\vee$-system.
%
Define a linear operator 
 $M\colon V\to V$ as
\begin{equation*}
M=\sum_{\beta\in \mathcal{A}}c_{\beta}\beta\otimes\beta,
\end{equation*}
that is, $M(u)=\sum_{\beta\in \mathcal{B}}c_{\beta}(\beta,u)\beta$ for any $u\in V$.  
%
%
%
%
The following statement takes place.
\begin{lemma}\label{decomposition of V.FF}
Let $(\mathcal{A},c)$ be a Euclidean trigonometric $\vee$-system.
Assume that the linear span $\langle \mathcal{A} \rangle = V.$
Then
\begin{enumerate}
\item
Any $\alpha\in \mathcal{A}$ is an eigenvector of $M$,  
%
\item
The vector space $V$ can be decomposed as
\begin{equation}\label{V-decomposition.FF}
V=V_{1}\oplus V_{2}\oplus\dots\oplus V_{k},\quad k\in \mathbb{N},
\end{equation}
where $M|_{V_{i}}=\lambda_{i}I,\,\lambda_{i}\in \mathbb{C},$ and $I$ is the identity operator, and $\lambda_i \neq \lambda_j$ for $i\neq j$. 
\end{enumerate}
\end {lemma}
The proof of Lemma \ref{decomposition of V.FF} is similar to the proof in \cite{Maali+Misha 2021} for the trigonometric $\vee$-system case (see also \cite{Maali.2022}).

Since $\mathcal{A}\subset V =V_{1}\oplus 
\dots\oplus V_{k},$
then $\mathcal{A}$ can be represented as a disjoint union
\begin{equation*}
\mathcal{A}=\mathcal{A}_{1}\sqcup 
\dots\sqcup \mathcal{A}_{k},    
\end{equation*}
where $\mathcal{A}_i \coloneqq \mathcal{A} \cap V_{i} \subset V_{i}$.
The following two lemmas relate the strings of vectors in $\mathcal{A}$ and its components $\mathcal{A}_i$.
\begin{lemma}\label{series in components of A.FF}
Let $\mathcal{A}$ be a Euclidean trigonometric $\vee$-system. 
Let $\alpha \in \mathcal{A}$ be such that
$\alpha\in V_{i}$ for some $i$. 
Consider an $\alpha$-string $\Gamma_{\alpha}^s$ in $\mathcal{A}_i$ and let $\beta \in \Gamma_{\alpha}^s.$
Then $\Gamma_{\alpha}^s \subset V_{i}$ or
$\Gamma_{\alpha}^s\subseteq\{\pm\beta\}.$ 
\end{lemma}
\begin{proof}
For $\beta \in \Gamma_{\alpha}^s$ we have two possible cases. 

\textit{Case (i)}  $\beta\in V_{i}.$
Then for any $\gamma \in \Gamma_{\alpha}^s$ we have that 
$\gamma=m\alpha +\varepsilon \beta \in V_{i} $  for some $m\in\mathbb{Z}$ and $\varepsilon=\pm 1$.
%
Hence $\Gamma_{\alpha}^s \subset V_{i}.$
%
%

\textit{Case (ii)}  $\beta\notin V_{i}.$ Hence
$\beta\in V_{j}$ for some $j\neq i.$
Then for any $  \gamma \in  \Gamma_{\alpha}^s$ we have that $\gamma \in V_{i} $ or 
$\gamma \in V_{j} $ since decomposition \eqref{V-decomposition.FF} is the direct sum.
Note that
$\gamma\notin V_{i}$
as otherwise
we will have  $\beta=m\alpha + \varepsilon\gamma \in V_{i},  $ for some $m\in\mathbb{Z}$ and $\varepsilon =\pm 1,$
which is a contradiction.
Note also that $\gamma \notin V_{j}$ unless $\gamma=\pm \beta$ as otherwise we have
$m\alpha=\beta +\varepsilon \gamma \in V_{j}$ for some $m\in \mathbb{Z}$ and $\varepsilon =\pm 1$,
which is a contradiction.
Hence $\Gamma_{\alpha}^s \subseteq\{\pm \beta \}$.
\end{proof}
\begin{lemma}\label{seriesrelations. Components of A.FF}
%
Let $\alpha,\beta\in \mathcal{A}_i.$ 
Let ${^{\mathcal{A}}\Gamma_{\alpha}^{s}},\,
{^{\mathcal{A}_i}\Gamma_{\alpha}^t}$ be the $\alpha$-strings in $\mathcal{A}$ and $\mathcal{A}_i$ respectively containing $\beta.$ Then the set ${^{\mathcal{A}}\Gamma_{\alpha}^{s}}$ is equal to the set ${^{\mathcal{A}_i}\Gamma_{\alpha}^t}$. 
\end{lemma}
\begin{proof}
Let $\gamma \in {^{\mathcal{A}}\Gamma_{\alpha}^{s}}.$ 
Then $\gamma = m\alpha +\varepsilon \beta \in \mathcal{A}_i,$ 
for some $m\in\mathbb{Z}$ and $\varepsilon =\pm 1$. 
Thus $\gamma \in {^{\mathcal{A}_i}\Gamma_{\alpha}^t} $ by the maximality of ${^{\mathcal{A}_i}\Gamma_{\alpha}^t}$. Hence ${^{\mathcal{A}}\Gamma_{\alpha}^{s}}\subseteq  {^{\mathcal{A}_i}\Gamma_{\alpha}^t}$.
The opposite inclusion is obvious. Therefore ${^{\mathcal{A}_i}\Gamma_{\alpha}^t}={^{\mathcal{A}}\Gamma_{\alpha}^{s}}$.
%
\end{proof}
%
%
Note that the operator $M$ is symmetric: $(M(u),v)=(u,M(v))$ for any $u,v\in V$. 
Hence its eigenspaces are orthogonal.
\begin{proposition}\label{orthogonal vectors in V decompos}
We have
$(u,v)=0$
for any $u \in V_i$ and $v \in V_j$ such that $i\neq j$.
\end{proposition}

The following statement takes place.
\begin{lemma}\label{restricted bilinear.non-deg}
Restriction $(\cdot,\cdot)_i$ of the bilinear form $(\cdot,\cdot)$ onto the subspace $V_i$ is non-degenerate.
\end{lemma}
\begin{proof}
Suppose that $v \in V_i$ satisfies $(v,u)_i =0 $ for all $u\in V_i$.
By Proposition \ref{orthogonal vectors in V decompos} we have $(v,u)=0$ for all $ u \in V $. Hence $v=0$ since $(\cdot,\cdot)$ is non-degenerate.
\end{proof}
%
%
The following statement relates the Euclidean trigonometric $\vee$-systems and the trigonometric $\vee$-systems.
\begin{theorem}\label{ierrEucTri is Trig}
If $\mathcal{A}$ is a Euclidean trigonometric $\vee$-system then the subsystem $\mathcal{A}_i =\mathcal{A} \cap V_i$ is well-distributed in the subspace $V_i$ with the bilinear form $(\cdot, \cdot)_i$ for all $i$.
Furthermore, if the bilinear form 
$$G_{\mathcal{A}_i ,c}(u,v)=\sum_{\alpha \in \mathcal{A}_i}c_{\alpha}(\alpha,u)(\alpha,v), \quad u,v \in V_i$$
is non-degenerate on $V_i$ (equivalently,  $G_{\mathcal{A}_i,c}$ is non-zero), then $\mathcal{A}_i$ is a trigonometric $\vee$-system.
Moreover, $\mathcal{A}$ is a trigonometric $\vee$-system if the form $G_{\mathcal{A},c}$ is non-degenerate.
%
%
%
\end{theorem}
\begin{proof}
By Lemma \ref{decomposition of V.FF} we have $M|_{V_i}= \lambda_i I$. 
Hence for any $u\in V_i$ and $v \in V$ we have
\begin{equation}\label{GA on V_i}
G_{\mathcal{A},c}(u,v) =(M(u),v)=\lambda_i (u,v).
\end{equation}
Note also that by Proposition \ref{orthogonal vectors in V decompos} 
we have 
that
\begin{equation}\label{GA in A and A_i}
  G_{\mathcal{A}_i,c}(u,v)
= G_{\mathcal{A},c}(u,v).
\end{equation}
Thus the subsystem $\mathcal{A}_i$ is well-distributed in the subspace $V_i$.

Let us now assume that $G_{\mathcal{A}_i,c}$ is non-degenerate on $V_i$, that is $\lambda_i \neq 0$.
Let $\alpha \in \mathcal{A}_i$. Consider an $\alpha$-string $\Gamma_{\alpha}^t$ in $\mathcal{A}_i.$
Then by Lemmas \ref{series in components of A.FF}, \ref{seriesrelations. Components of A.FF} and formulas \eqref{GA on V_i}, \eqref{GA in A and A_i} we have 
\begin{equation}\label{v-cond.1}
\sum_{\beta\in \Gamma_{\alpha}^t} c_\beta G_{\mathcal{A}_i,c}(\alpha,\beta) \alpha \wedge \beta =\lambda_i \sum_{\beta\in \Gamma_{\alpha}^t} c_\beta (\alpha,\beta) \alpha \wedge \beta = 0
\end{equation}
since $\mathcal{A}$ is a Euclidean trigonometric $\vee$-system.
This proves that $\mathcal{A}_i$ is a trigonometric $\vee$-system.
Finally, for $\alpha \in \mathcal{A}_i$ let us consider its $\alpha$-string $^{\mathcal{A}}\Gamma_{\alpha}^s$ in $\mathcal{A}$.
If $^{\mathcal{A}}\Gamma_{\alpha}^s \subset V_i$ then 
$\sum_{\beta \in {^{\mathcal{A}}\Gamma_{\alpha}^s }} c_{\beta}G_{\mathcal{A},c}(\alpha, \beta)\alpha \wedge \beta =0$ by Lemma \ref{seriesrelations. Components of A.FF} and \eqref{v-cond.1}.
If $^{\mathcal{A}}\Gamma_{\alpha}^s \not\subset V_i$ then $^{\mathcal{A}}\Gamma_{\alpha}^s \subseteq \{\pm \beta  \}$ for some $\beta \in V_j, j\neq i$, by Lemma \ref{series in components of A.FF}.
Then $G_{\mathcal{A},c}(\alpha,\beta)=\lambda_i (\alpha ,\beta)=\lambda_j (\alpha,\beta)$ by \eqref{GA on V_i}.
Hence $G_{\mathcal{A},c}(\alpha,\beta)=0$ and the trigonometric $\vee$-system condition holds.
\end{proof}
Let $U\subseteq{V}$ be a linear subspace such that $\langle \mathcal{A} \cap U \rangle =U$. The following statement takes place.
%
%
\begin{proposition}\label{well-dis on 2-dim}
Let $(\mathcal{A},c)$ be a Euclidean trigonometric $\vee$-system.
Then the set of vectors $\mathcal{A}\cap U$ 
with the multiplicity function $c|_{\mathcal{A}\cap U}$
is well-distributed in $U$ or the system $\mathcal{A}\cap U$ is reducible.
%
%
%
\end{proposition}
\begin{proof}
%
Define a linear operator $M_{U}\colon U \to U$ by
$$M_{U}\coloneqq\sum_{\beta\in \mathcal{A}\cap U}c_{\beta}\beta\otimes \beta.$$ 
Let $\alpha \in \mathcal{A}\cap U$. 
Let us sum up the Euclidean trigonometric $\vee$-condition \eqref{Euclid.Trig.condition} over $\alpha$-strings which belong to the subspace $U$. Then
\begin{equation*}
\sum_{\beta\in \mathcal{A}\cap U}c_{\beta}(\beta,\alpha)\beta
=M_U(\alpha)
=\lambda\alpha
\end{equation*}
for some $\lambda=\lambda(\alpha).$
Suppose that $\mathcal{A}\cap U$ is irreducible. Then $\lambda$ does not depend on $\alpha$ and $M_{U}=\lambda I$. Therefore 
$$\sum_{\beta\in \mathcal{A}\cap U}c_{\beta}(\beta,u)(\beta,v)=(M_{U}(u),v)=\lambda(u,v),$$
and the pair $(\mathcal{A}\cap U, c|_{\mathcal{A}\cap U})$ is well-distributed.
%
%
%
%
%
%
 %
\end{proof}
Suppose that $G_{\mathcal{A},c}$ is non-degenerate and define the vector $\alpha^{\vee}\in V$ by the relation
\begin{equation}\label{dual vector}
    G_{\mathcal{A},c}(\alpha^{\vee},x)=(\alpha,x)
\end{equation}
for all $x\in V$.
Now assume that $\alpha \in V_i$ in which case we also have $\alpha^{\vee}\in V_i$ for some $i$. Then by Lemma \ref{decomposition of V.FF} we have 
\begin{equation}\label{GA for the dual vector}
G_{\mathcal{A},c}(\alpha^{\vee},x) 
=(M_U(\alpha^{\vee}),x)
=\lambda_i (\alpha^{\vee},x).
\end{equation}
Hence from relations \eqref{dual vector}, \eqref{GA for the dual vector} we have 
that $\alpha^{\vee}={\lambda_i}^{-1} \alpha$. 
Therefore if the pair $(\mathcal{A},c)$ satisfies conditions \eqref{Trig.condition.Vector} then for all $\alpha$-strings $\Gamma_{\alpha}^{s}$
we have
%
%
\begin{equation*}
   \sum_{\beta\in \Gamma_{\alpha}^{s} }c_{\beta}G_{\mathcal{A},c}(\alpha^{\vee}, \beta^{\vee})\alpha \wedge \beta
   =
   \lambda_{i}^{-2} \sum_{\beta\in \Gamma_{\alpha}^{s} }c_{\beta}G_{\mathcal{A},c}(\alpha, \beta)\alpha \wedge \beta=0.
\end{equation*}
These conditions coincide with the definition of the trigonometric $\vee$-system given in \cites{Misha2009} so the two definitions are equivalent. 
Let us now introduce an inner product $\langle \cdot ,\cdot\rangle$ on $V$ as
\begin{equation}\label{new inner product}
\langle u,v \rangle\coloneqq G_{\mathcal{A},c}(u^{\vee},v^{\vee}), \quad u,v \in V .
\end{equation}

The following statement is immediate.
\begin{proposition}
Let $(\mathcal{A},c)$ be a trigonometric $\vee$-system.
Then  $(\mathcal{A},c)$ is a Euclidean trigonometric $\vee$-system with respect to the bilinear form \eqref{new inner product}.
\end{proposition}
\subsection{Relation with complex Euclidean $\vee$-systems}\label{section.Euclid+Complex.1}
%
%
Following \cite{Misha&Veselov 2008}, let us recall the notion of the (rational) complex Euclidean $\vee$-system. 
Let $V$ be a complex vector space with a non-degenerate bilinear form $(\cdot ,\cdot).$ 
Let $\mathcal{A}\subset V$ be a finite set of vectors. 
Consider the canonical form
$$G_{\mathcal{A}}^{r}(x,y)=\sum_{\alpha\in \mathcal{A}}(\alpha, x)(\alpha , y), \quad x,y \in V.$$
Suppose that $G_{\mathcal{A}}^{r}$ is proportional to the form $(\cdot,\cdot)$.
Let $\pi \subseteq{V}$ be a two-dimensional subspace such that $\langle \mathcal{A} \cap \pi \rangle =\pi$.
 $\mathcal{A}$ is said to be a \textit{ (rational) complex Euclidean $\vee$-system} if for any such $\pi$ the subsystem $\mathcal{B}=\mathcal{A}\cap \pi$ is reducible or the corresponding form ${G_{\mathcal{B}}^r}|_{\pi}$ is proportional to $(\cdot,\cdot)|_{\pi}$.
%
 
%
%
\begin{proposition}
Let $F=\sum_{\alpha\in \mathcal{A}}(\alpha,x)^2 \log (\alpha,x)$.
Suppose that $F$ satisfies the commutativity equations and $\mathcal{A}$ is irreducible. Then $\mathcal{A}$ is a (rational) complex Euclidean $\vee$-system. Moreover, if for a two-dimensional plane $\pi \subset V$ the subsystem $\mathcal{A}\cap \pi$ is reducible then the corresponding directions are orthogonal with respect to the form $(\cdot ,\cdot)$.
\end{proposition}
Proof is similar to \cite{Misha&Veselov 2008}. The substitution of $F$ into the commutativity equations gives the condition
$$\sum_{\alpha,\beta \in \mathcal{A}}\frac{(\alpha,\beta)}{(\alpha,x)(\beta ,x)}(\alpha \wedge \beta)(\alpha \wedge \beta)=0.$$
It implies that for any two-dimensional plane $\pi$
$$\sum_{\beta \in \pi\cap {\mathcal A}}(\alpha ,\beta)\alpha \wedge \beta =0. $$

The following statement relates the Euclidean trigonometric $\vee$-system and the (rational) complex Euclidean $\vee$-system.
\begin{proposition}
Let $(\mathcal{A},c)$ be an irreducible Euclidean trigonometric $\vee$-system.   
Then the set of vectors $\sqrt{c_\alpha} \alpha,\, \alpha \in \mathcal{A}$, is a (rational) complex Euclidean $\vee$-system.
\end{proposition}
\begin{proof}
Firstly, since $\mathcal{A}$ is irreducible then by Lemma \ref{decomposition of V.FF} we have
$\mathcal{A}=\mathcal{A}_1 \subset V=V_1$ and $M|_{V_1}=\lambda_1 I.$
Then by Theorem \ref{ierrEucTri is Trig} we have that $\mathcal{A}_1 =\mathcal{A}$ is well-distributed. 

Secondly, by Proposition \ref{well-dis on 2-dim} we have that any subsystem of $\mathcal{A}$ is well-distributed or reducible, which implies the statement. 
\end{proof}

%
\section{Restricted solutions of commutativity equations}\label{section.restrictions of Euclid.trig}
In this Section we apply the restriction procedure to a given solution to the commutativity equations.
This gives new solutions of the commutativity equations.
%

Let 
$ \mathcal{B}=\mathcal{A} \cap W$
be a subsystem of $\mathcal{A}$ for some  $n$-dimensional linear subspace $ W=\langle \mathcal{B} \rangle \subset V.$ 
Define
\begin{equation*}
W_{\mathcal{B}} \coloneqq \{x \in V \colon (\beta,x)=0 \quad \forall \beta \in \mathcal{B}\}.   
\end{equation*}

Let $(\cdot,\cdot)_{\mathcal{B}}$ be the restriction of $(\cdot,\cdot)$ on $W_{\mathcal{B}}, $ and assume that it is non-degenerate.
%
%
 %
%
%
Let $S \subset \mathcal{B}$ 
%
be a basis of $W$.
Let $f_1 ,\dots, f_n$ be an orthonormal basis of the space $W_{\mathcal{B}}$, 
and let ${\xi=(\xi^1,\dots,\xi^n)}$ be the corresponding orthonormal coordinates in $W_{\mathcal{B}}$.
Define $M_{\mathcal{A}}=V \setminus \bigcup_{\alpha \in \mathcal{A}} \Pi_{\alpha},$  and
$M_{\mathcal{B}}=W_{\mathcal{B}} \setminus \bigcup_{\alpha \in \mathcal{A}\setminus \mathcal{B}} \Pi_{\alpha},$ 
where $\Pi_{\alpha}=\{x \in V\colon (\alpha,x)=0 \}.$
The following statement shows that the class of solutions of commutativity equations corresponding to Euclidean $\vee$-systems is closed under the restrictions.
\begin{theorem} \label{restricted system and commutativity equations}
Assume that prepotential \eqref{trig.solution. general form 2} satisfies commutativity equations \eqref{FF=FF.1}. 
%
%
Suppose that  $C_{\delta}^{\alpha_{0}}\neq 0$ for any $\alpha \in S, \alpha_{0}\in \delta_{\alpha},\, \delta\subseteq \delta_{\alpha}.$
Then 
the prepotential 
\begin{equation} \label{F_B solution for FF}
F_{\mathcal{B}}=\sum_{\alpha \in \mathcal{A}\setminus \mathcal{B}}c_{\alpha}f(({\alpha},\xi)),\quad \xi \in M_{\mathcal{B}},
\end{equation}
%
satisfies the commutativity equations
\begin{equation*}
 (F_{\mathcal{B}})_i (F_{\mathcal{B}})_j=(F_{\mathcal{B}})_j (F_{\mathcal{B}})_i , \quad i,j=1,\dots , n, 
\end{equation*}
where $(F_{\mathcal{B}})_i$ is the $n\times n$ matrix with entries
$$((F_{\mathcal{B}})_i)_{pq}=(F_{\mathcal{B}})_{ipq} = \frac{\partial^{3} F_{\mathcal{B}}}{\partial \xi^i \partial \xi^p \partial \xi^q}.$$
\end{theorem}
\begin{proof}
First for any $u=(u^1 ,\dots, u^N),\, v=(v^1 ,\dots, v^N)\in V$ 
let us consider the vector fields $\partial_{u}=u^i \partial_{x^i},\, \partial_{v}=v^i \partial_{x^i} \in T_x M_{\mathcal{A}}.  $
We define the following multiplication on the tangent space $T_x M_{\mathcal{A}}$:
\begin{equation}\label{product for FF}
    \partial_{u}\ast \partial_v =u^i v^j 
     F_{ijk} \partial_{x^k}.
\end{equation}
It is easy to check that the associativity of the multiplication $\ast $ is equivalent to the commutativity equations \eqref{FF=FF.1}.
From the formula \eqref{trig.solution. general form 2} we have
$$F_{ijk}=\sum_{\alpha\in \mathcal{A}}c_{\alpha}(\alpha , f_i) (\alpha , f_j) (\alpha, f_k) \cot{(\alpha,x)}.$$
Hence multiplication \eqref{product for FF}
takes the form
\begin{equation}\label{prod 2}
     \partial_{u}\ast \partial_v =\sum_{\alpha\in \mathcal{A}}c_{\alpha}(\alpha,u)(\alpha,v)\cot{(\alpha,x)}
     \partial_{\alpha}.
\end{equation}
By identifying $V \cong T_x V$, we have 
\begin{equation*}
     u\ast v =\sum_{\alpha\in \mathcal{A}}c_{\alpha}(\alpha,u)(\alpha,v)\cot{(\alpha,x)}
     {\alpha}.
\end{equation*}
Consider now a point $x_{0} \in M_{\mathcal{B}}$ and two tangent vectors $u_{0}, v_{0} \in T_{x_{0}}M_{\mathcal{B}}.$ We extend vectors $u_{0}$ and $v_{0}$ to two local analytic vector fields $u(x), v(x)$ in the neighbourhood $U$ of $x_{0}$ that are tangent to the subspace $W_{\mathcal{B}} $ at any point $x \in M_{\mathcal{B}} \cap U$ such that $u_{0}=u(x_{0})$ and  $v_{0}=v(x_{0})$. 
The proof of the next lemma is similar to the proof of \cite{Maali+Misha 2021}*{Lemma 4.1} 
(see also \cite{Misha&Veselov 2007}*{Lemma 1}
for the rational case, and  \cite{MGM 2020} for the trigonometric case).
\begin{lemma}\label{the limit of * for FF}
The limit of the product $u(x) \ast v(x)$ exists when vector $x$ tends to $x_{0}\in  M_{\mathcal{B}} $ and it satisfies
\begin{equation*}
u_{0} \ast v_{0} =\sum_{\alpha \in \mathcal{A}\setminus \mathcal{B}} c_{\alpha}(\alpha,u_{0})(\alpha,v_{0})\cot (\alpha, x_{0})\alpha.
\end{equation*}
In particular, the product $u_{0} \ast v_{0}$ is determined by vectors $u_{0}$ and $v_{0}$ only. 
\end{lemma}
The following lemma holds and it shows that multiplication \eqref{prod 2} is closed on the tangent space $T_{x_{0}}M_{\mathcal{B}}$. 
\begin{lemma}\label{closedalgebra FF}
Let $u,v \in T_{x_0}M_{\mathcal{B}}$ where $x_0 \in M_{\mathcal{B}}.$ Then $u \ast v \in T_{x_0}M_{\mathcal{B}}.$ 
\end{lemma}
The proof of Lemma \ref{closedalgebra FF} is similar to the proof of \cite{Maali+Misha 2021}*{Lemma 4.2}.
It uses an argument analogues to \cite{Maali+Misha 2021} 
which claims that the following identity holds for any $a,b\in V$ if $\tan (\alpha,x)=0$:
$$\sum_{\beta \in \mathcal{A}\setminus \delta_{\alpha} }c_{\beta} (\alpha,\beta)\cot (\beta,x)B_{\alpha , \beta}(a \otimes b)\alpha \wedge \beta =0.$$
%
%
Then for $u,v \in T_{x_0}M_{\mathcal{B}}$,  $x_0\in M_{\mathcal{B}},$ the product \eqref{prod 2} takes the form
\begin{equation*}
     \partial_{u}\ast \partial_v =\sum_{\alpha\in \mathcal{A}\setminus \mathcal{B}}c_{\alpha}({\alpha},u)({\alpha},v)\cot{({\alpha},x_0)}
     \partial_{\alpha}.
\end{equation*}
By using the orthonormal basis $f_1 ,\dots, f_n$ of $W_{\mathcal{B}}$ we rearrange $\partial_{{\alpha}}$ as
$$\partial_{\alpha}=\sum_{k=1}^{n}({\alpha},f_k)\partial_{f_k}.$$
Hence for $x_0=\xi=\sum_{i=1}^{n}\xi^i f_i$ we have 
\begin{align}\label{prod 5}
    \partial_{f_i} \ast \partial_{f_j}&=\sum_{\alpha\in \mathcal{A}\setminus \mathcal{B}}\sum_{k=1}^{n}c_{\alpha}({\alpha},f_i)({\alpha},f_j)({\alpha},f_k) \cot{({\alpha},\xi)}\partial_{f_k}\nonumber\\
    &=\sum_{k=1}^{n} \widetilde{F}_{ijk}\partial_{f_k},\quad i,j=1,\dots, n,
\end{align}
where $\widetilde{F}(\xi)=\sum_{\alpha\in \mathcal{A}\setminus \mathcal{B}}c_{\alpha} f(\alpha,\xi)= F_{\mathcal{B}}.$
Now multiplication \eqref{prod 5} is associative and it is easy to check that its associativity is equivalent to the commutativity equations 
$$\widetilde{F}_i \widetilde{F}_j=\widetilde{F}_j \widetilde{F}_i, \quad i,j=1,\dots, n.$$
%
\end{proof}
As an application of Theorem \ref{restricted system and commutativity equations} we get the following solutions  of the commutativity equations starting from a solution for the root system $BC_N$ \cite{MGM 2020}.
%
Let $q,r,s \in \mathbb{C},\, \underline{m}=(m_1 ,\dots, m_n)\in \mathbb{(C^\times)}^n$.
Suppose
\begin{equation*}
r = -8s- 2q(N- 2)
\end{equation*}
with $N=\sum_{i=1}^n m_i$.
 Define the configuration $BC_n (q,r,s; \underline{m})\subset \mathbb{C}^n$ consisting of the following vectors $\alpha$ with the multiplicities $c_\alpha$:
\begin{align}\label{rest of BCn.orthogonal}
& m_i^{-1/2} e_i, \quad \text{with multiplicity} \quad r m_i, \quad 1 \leq i \leq n,\nonumber \\
&  2 m_i^{-1/2} e_i, \quad \text{with multiplicity} \quad s m_i + \frac{1}{2} qm_i(m_i -1),\quad 1 \leq i \leq n,\nonumber \\
& m_i^{-1/2} e_i \pm m_j^{-1/2} e_j, \quad \text{with multiplicity} \quad q m_i m_j, \quad 1 \leq i < j \leq n.
\end{align}
Note that in the case  $m_i =1$ for all $i=1,\dots, n$ configuration \eqref{rest of BCn.orthogonal} reduces to the positive half $BC_n^+$ of the root system $BC_n$. 
%
Consider the function 
\begin{equation}\label{F for restictions of BCn}
  \widetilde{F}=\sum_{\alpha\in BC_n(q,r,s;\underline{m})} c_{\alpha} f((\alpha ,{x})), \quad {x}\in \mathbb{C}^n. 
\end{equation}
It follows from Theorem \ref{restricted system and commutativity equations} and it was shown earlier in \cite{MGM 2020} that function \eqref{F for restictions of BCn} satisfies the commutativity equations
$$ \widetilde{F}_i  \widetilde{F}_j =  \widetilde{F}_j  \widetilde{F}_i \quad i,j=1,\dots, n.$$
%
 
%

Consider now the positive half $F_4^+$ of the root system $F_4$ given by  
\begin{equation}\label{F4.half positive}
F_4^{+}=\{ e_i\, (1\leq i\leq 4),\,\, e_i \pm e_j \, (1\leq i< j\leq 4),\,\, \frac{1}{2}(e_1\pm e_2 \pm e_3 \pm e_4 )  \}.
\end{equation}
Let $r$ be the multiplicity of short roots and let $q$ be the multiplicity of long roots. 
A basis of simple roots consists of
${\alpha_1={e_2-e_3}, \alpha_2={e_3-e_4}, \alpha_3={e_4}}$,
${\alpha_4=\frac{1}{2}(e_1-e_2-e_3-e_4)}.$

Up to an orthogonal transformation there are two projected systems in dimension three and four projected systems in dimension two. 
Firstly, we give details of the three-dimensional projections of $F_4^+$.
There are two different projections of $F_4^+$ along the root system $A_1$.
The first one $(F_4,A_1)_{1}$ is obtained by projecting to the hyperplane $\alpha_3=0.$
The second one $(F_4,A_1)_{2}$ is obtained by projecting to the hyperplane  $\alpha_2=0.$
Hence we have the following three-dimensional projected systems of the positive root system $F_4^+$.
\begin{itemize}
    \item The projected system $(F_4 , A_1)_1$
    consists of the following vectors:
\begin{align*}
& e_i, \quad  \text{with multiplicity} \quad r+2q ,\quad 1 \leq i \leq 3,\\
& e_i\pm e_j, \quad \text{with multiplicity} \quad q ,\quad 1 \leq i<j \leq 3,\\
& \frac{1}{2}(e_1 \pm e_2 \pm e_3), \quad \text{with multiplicity} \quad 2r.
\end{align*}
\item The projected system $(F_4 , A_1)_2$  
consists of the following vectors (after doing a change of variables and renaming vectors):
\begin{align*}
& e_1, e_2, \quad  \text{with multiplicity} \quad r ,\\
& \sqrt{2}e_3, \quad  \text{with multiplicity} \quad q ,\\
& \frac{\sqrt{2}}{2}e_3, \quad  \text{with multiplicity} \quad  2r  ,\\
& e_1\pm e_2, \quad \text{with multiplicity} \quad q ,\\
& \frac{1}{2} (e_1\pm e_2), \quad \text{with multiplicity} \quad 2r ,\\
& e_1 \pm \frac{\sqrt{2}}{2}e_3,\, e_2 \pm \frac{\sqrt{2}}{2}e_3, \quad  \text{with multiplicity} \quad 2q ,\\
& \frac{1}{2}(e_1 \pm e_2 \pm \sqrt{2} e_3), \quad \text{with multiplicity} \quad r.
\end{align*}
\end{itemize}

Secondly, we give details of the two-dimensional projections of $F_4^+$.
There are two different projections of $F_4^+$ along the root system $A_2$.
The first one $(F_4,A_2)_{1}$ is obtained by projecting to the plane $\alpha_1=\alpha_2=0.$
The second one $(F_4,A_2)_{2}$ is obtained by projecting to the plane $\alpha_3=\alpha_4=0.$
There is also a projected configuration $(F_4, A_1^2)$ along the subsystem $A_1 \times A_1$ to the plane $\alpha_1 =\alpha_3=0$, and there is a projected configuration $(F_4, B_2)$ along the subsystem $B_2$ to the plane $\alpha_2 =\alpha_3 =0$.
These configurations have the following explicit form 
 
 \begin{itemize}
     \item 
The configuration $(F_4 ,A_2)_1$ consists of vectors $\alpha$ with the corresponding multiplicities $c_\alpha$ given as follows:
\begin{align*}
& e_1, \quad  \text{with multiplicity} \quad r,\nonumber \\
& \frac{1}{\sqrt{3}}e_2, \quad  \text{with multiplicity} \quad 3r ,\nonumber\\
& \frac{2}{\sqrt{3}}e_2, \quad  \text{with multiplicity} \quad 3q ,\nonumber\\
& e_1 \pm \frac{1}{\sqrt{3}} e_2, \quad  \text{with multiplicity} \quad 3q  ,\nonumber\\
&\frac{1}{2} (e_1\pm \frac{1}{\sqrt{3}} e_2), \quad \text{with multiplicity} \quad 3r ,\nonumber\\
&\frac{1}{2} (e_1\pm \frac{3}{\sqrt{3}} e_2), \quad \text{with multiplicity} \quad r. 
\end{align*}
\item
The configuration $(F_4 ,A_1^2)$ consists of vectors $\alpha$ with the corresponding multiplicities $c_\alpha$ given as follows:
\begin{align*}
& e_1, \quad  \text{with multiplicity} \quad r+2q ,\nonumber \\
& \sqrt{2}e_2, \quad  \text{with multiplicity} \quad q ,\nonumber\\
& \frac{1}{2}e_1, \quad  \text{with multiplicity} \quad 4r ,\nonumber\\
& \frac{\sqrt{2}}{2}e_2, \quad  \text{with multiplicity} \quad  2(r+2q)  ,\nonumber\\
& e_1\pm \frac{1}{\sqrt{2}} e_2, \quad \text{with multiplicity} \quad 2q ,\nonumber\\
& \frac{1}{2} (e_1\pm \frac{2}{\sqrt{2}} e_2), \quad \text{with multiplicity} \quad 2r.
\end{align*}
 \end{itemize} 
%
Configurations $(F_4 , A_2)_2$ and $ (F_4 , B_{2})$ are equivalent to the root systems $G_2$ and $BC_2$, respectively, the corresponding solutions of the commutativity equations were found in \cite{George+Misha 2019}. 
Theorem \ref {restricted system and commutativity equations} gives new solutions of the commutativity equations which are listed in the next statement.
\begin{theorem}\label{commutativity for rest.of F4}
Let $({\mathcal{A}},c)$ be one of the configurations  $(F_4 , A_1)_1 ,\, (F_4 , A_1)_2, \, (F_4 , A_2)_1 , \, (F_4 , A_1^2)$ described above. 
Then the function $F=\sum_{\alpha\in {\mathcal{A}}} c_{\alpha}f((\alpha,x))$ satisfies the commutativity equations, where $x\in \mathbb{C}^3$ for the first two configurations and $x\in \mathbb{C}^2$ for the last two configurations and parameters $r, q$ satisfy the condition $r=-2q$ or $r=-4q$.
%
\end{theorem}


\section{Commutativity equations and WDVV equations}\label{CommEqAndWDVV}

In this Section we investigate the relation between the commutativity equations and the WDVV equations. 
%
Let $V\cong \mathbb{C}^N, N\geq 2$. Let $F=F(x^1, \dots, x^N)$ be a function on $V$.
We recall that it was proven in \cite{Martini+Gragert 1999} (see also \cite{MMM.2000}) that the generalized WDVV equations
\begin{equation*}
{F}_i {F}_{j}^{-1} {F}_k= {F}_k {F}_{j}^{-1} {F}_i, \quad i,j,k=1, \dots,N 
\end{equation*}
%
%
%
can be written equivalently in the form  
\begin{equation}\label{WDVV with metric B form}
     {F}_i B^{-1}  {F}_j=  {F}_j B^{-1}  {F}_i, \quad i,j=1,\dots, N,
 \end{equation}
 where $B$ is any non-degenerate matrix of the form $B=\sum_{i=1}^{N}A^k F_k$ for some functions $A^k$. 
%
If the matrix $B$ happens to be a multiple of the identity matrix 
then WDVV equations \eqref{WDVV with metric B form} reduce to the commutativity equations
\begin{equation}\label{FF.ch4}
     {F}_i  {F}_j=  {F}_j   {F}_i, \quad i,j=1,\dots, N.
 \end{equation}
A natural question to investigate is when there exists such a linear combination $B$. 
We give an answer in this Section.

Let us assume that a function $F=F(x^1, \dots, x^N)$ satisfies the commutativity equations \eqref{FF.ch4}.
Let us denote by $[F_i,F_j]_{(a,b)}$ the $(a,b)$-entry of the commutator $[F_i,F_j]$ or, more explicitly,
\begin{equation}\label{ab entry}
    [F_i,F_j]_{(a,b)}=\sum_{m=1}^N (F_{iam}F_{jbm}-F_{ibm} F_{jam}).
\end{equation}
%
The equality $[F_i, F_k]_{(i,j)}=0$ implies that
%
%
 \begin{align}\label{ij entry of FiFk}
     F_{ijk}F_{iii}=\sum_{m=1}^{N}F_{ijm}F_{ikm}- \sum_{ m\neq i}  F_{iim}F_{jkm}.   \end{align}
     %
     %
     %
 %
 %
 
 Observe the equality of matrix entries
$[F_a , F_b]_{(i,j)}=[F_i , F_j]_{(a,b)}$.
Introduce the notation
\begin{equation*}
    [F_i,F_j]_{(a,b)}^{\{ m \}}=F_{iam}F_{jbm}-F_{ibm}F_{jam},
\end{equation*}
where there is no summation over $m$ in the right-hand side.  
Define a matrix 
$B=(B_{ij})_{i,j=1}^{N}$ with the entries given as a linear combination of the third order derivatives of $F$:
\begin{equation}\label{general formula of the metric}
 B_{ij}=\sum_{k=1}^{N} A^{k} F_{kij},   
\end{equation}
for some functions $A^k=A^k (x^1,\dots, x^N)$.
Now we will investigate when there exists such a combination $B$ so that equations \eqref{FF.ch4} imply equations \eqref{WDVV with metric B form}.
For that it is sufficient to deduce that the matrix $B$ is proportional to the identity.
%

 %
 %
 %
 %

 %
%
%
Fix $i_0\in \mathbb{N},\, 1\leq i_0 \leq N$.
Let $P$ be the $(N-1)\times N$ matrix $P=(F_{i_0 ij}), \, 1\leq i,j\leq N,\, i\neq i_0$.
Define 
\begin{equation}\label{A_k in N-dim}
A^{k} = (-1)^{k+1}\det P_k,
\end{equation}
where the matrix $P_k$ is obtained from the matrix $P$ by removing its $k$-th column.
%
%

%
The following statement takes place.
\begin{lemma}\label{property w.r.t FF. General entries}
For any function $F=F(x^1,\dots, x^N)$ which satisfies the commutativity equations \eqref{FF.ch4} the following relation holds
 \begin{equation*}
 \det 
 \begin{pmatrix}
 F_{a i_0 r} & F_{b i_0 r} & F_{c i_0 r}\\
  F_{a i_0 t} & F_{b i_0 t} & F_{c i_0 t}\\
   F_{art} & F_{brt} & F_{crt}
 \end{pmatrix}
 = - \sum_{m\neq t} \det
  \begin{pmatrix}
 F_{a i_0 r} & F_{b i_0 r} & F_{c i_0 r}\\
  F_{a i_0 m} & F_{b i_0 m} & F_{c i_0 m}\\
   F_{arm} & F_{brm} & F_{crm}
 \end{pmatrix},
 \end{equation*}
 where $1\leq a,b ,c \leq N, \, 1\leq r < t\leq N$, and $r,t \neq i_0$.
\end{lemma}
\begin{proof}
By applying the first row expansion and the commutativity equations we get
\begin{align*}
&\det
   \begin{pmatrix}
 F_{a i_0 r} & F_{b i_0 r} & F_{c i_0 r}\\
  F_{a i_0 t} & F_{b i_0 t} & F_{c i_0 t}\\
   F_{art} & F_{brt} & F_{crt}
 \end{pmatrix}
 = F_{a i_0 r} \det
  \begin{pmatrix}
  F_{b i_0 t} & F_{c i_0 t}\\
  F_{brt} & F_{crt}
 \end{pmatrix}
 -F_{b i_0 r} \det
  \begin{pmatrix}
  F_{a i_0 t}  & F_{c i_0 t}\\
  F_{art} & F_{crt}
 \end{pmatrix}  \\
& + F_{c i_0 r} \det
  \begin{pmatrix}
  F_{a i_0 t} & F_{b i_0 t} \\
   F_{art} & F_{brt}
 \end{pmatrix} 
 = F_{a i_0 r} [F_{i_0},F_r]_{(b,c)}^{\{t\}}- F_{b i_0 r} [F_{i_0},F_r]_{(a,c)}^{\{t\}}+ F_{c i_0 r} [F_{i_0},F_r]_{(a,b)}^{\{t\}}\\
& = - F_{a i_0 r}\sum_{m\neq t} [F_{i_0},F_r]_{(b,c)}^{\{m\}} 
+ F_{bi_0 r}\sum_{m\neq t} [F_{i_0},F_r]_{(a,c)}^{\{m\}}
- F_{c i_0 r}\sum_{m\neq t} [F_{i_0},F_r]_{(a,b)}^{\{m\}}\\
& = - \sum_{m\neq t} \det
  \begin{pmatrix}
 F_{a i_0 r} & F_{b i_0 r} & F_{c i_0 r}\\
  F_{a i_0 m} & F_{b i_0 m} & F_{c i_0 m}\\
   F_{arm} & F_{brm} & F_{crm}
 \end{pmatrix}.
  \end{align*}
\end{proof}
The following statement takes place.
\begin{lemma}\label{sinqular matrix for rt-entry}
Suppose that a function $F=F(x^1,\dots, x^N)$ satisfies the commutativity equations \eqref{FF.ch4}. 
Let $1\leq r<t \leq N,\, r,t \neq i_0$. Let $N\times N$ matrix $Q$ be obtained from the matrix $P$ by inserting the $i_0$-th row
$R_{i_0}=(F_{1rt}, \dots ,F_{Nrt })$. Then $\det Q =0$.
 
\end{lemma}
\begin{proof}
Let $D=\det Q$.
Let $R_i$ denote the $i$-{th} row in the matrix $Q$.
Let us perform the Laplace expansion of $D$ along the rows 
\begin{align*}
    R_{r}&=(F_{i_0 1r},F_{i_0 2 r}, \dots , F_{i_0 N r}),\\
     R_{t}&=(F_{i_0 1 t},F_{i_0 2 t}, \dots , F_{i_0 N t}),
\end{align*}
and the row $R_{i_0}$.
For subsets $S,T\subset [N]\coloneqq \{1,\dots, N   \}$ we define $Q_{ST}$ to be the submatrix of $Q$ defined by deleting rows $S$ and columns $T$.  
Let $I=\{r,t,i_0\},\, J=\{a,b,c \}$ for some $1\leq a<b<c\leq N.$
%
By applying the Laplace expansion to $Q$ we get

 \begin{equation}\label{Mrt.v1}
D=\varepsilon  \sum_{J} \sigma_J \det Q_{I J} \det
 \begin{pmatrix}
 F_{a i_0 r} & F_{b i_0 r} & F_{c i_0 r}\\
  F_{a i_0 t} & F_{b i_0 t} & F_{c i_0 t}\\
   F_{art} & F_{brt} & F_{crt}
 \end{pmatrix},
\end{equation}
where $\varepsilon = \pm 1$ is determined by the relative order of $r,t$ and $i_0$, and $\sigma_J =(-1)^{s}$ with $s=\sum_{i\in I}i +  \sum_{j\in J} j$. 
%
%
Now by Lemma \ref{property w.r.t FF. General entries}
the determinant \eqref{Mrt.v1} can be rewritten as
\begin{align*}\label{Mrt.v2}
 D &=-\varepsilon \sum_{J}  \sigma_J \det Q_{I J} \Big(  \sum_{m\neq t} \det 
 \begin{pmatrix}
 F_{a i_0 r} & F_{b i_0 r} & F_{c i_0 r}\\
  F_{a i_0 m} & F_{b i_0 m} & F_{c i_0 m}\\
   F_{arm} & F_{brm} & F_{crm}
 \end{pmatrix}\Big)\nonumber\\
& = -\sum_{m\neq t} \det Q^{(m)} ,
\end{align*}
where the matrix $Q^{(m)}$ is obtained from the matrix $Q$ by replacing rows $R_{i_0}$ and $R_t$ as follows:
\begin{align*}
    R_{i_0}\to \widetilde{R}_{i_0}=(F_{1rm}, \dots, F_{Nrm}),\\
     R_{t}\to \widetilde{R}_{t}=(F_{i_0 1 m}, \dots, F_{i_0 N m}).
\end{align*}
Note that the row $\widetilde{R}_{t}$ is equal to the $m$-th row of the matrix $Q^{(m)}$. Hence $\det D=0$.
%
\end{proof}
The following statement takes place.
%
\begin{proposition}\label{non-diagonal metric.N-dim}
Assume that the function $F=F(x^1, \dots, x^N)$ satisfies the commutativity equations \eqref{FF.ch4}. Assume also that the rank of the matrix $P$ is $N-1$.
Then matrix $B$ with the entries given by formulae \eqref{general formula of the metric}, \eqref{A_k in N-dim} is diagonal. 
\end{proposition}
\begin{proof}
Let us assume that $i_0 =1$, the general case can be dealt with similarly.
Consider the system of linear equations
\begin{equation}\label{non-diagonal B1m}
 B_{1m}=\sum_{k=1}^{N} A^{k} F_{1km}=0,   
\end{equation}
for some functions $A^k=A^k (x^1,\dots, x^N)$.
The system \eqref{non-diagonal B1m} represents a homogeneous system of $N-1$ linear equations in variables $A^k$.
The assumption that the rank of the matrix  $P=(F_{1ij})$, where $ 2\leq i\leq N,\, 1\leq j\leq N$, is $N-1$ implies that the system \eqref{non-diagonal B1m} has a non-trivial solution which is unique up to proportionality.
%
%

Now, fix  $2 \leq s \leq N $. 
The direct substitution of the functions $A^k$ given by formula \eqref{A_k in N-dim} into the right-hand side of relation \eqref{non-diagonal B1m} gives a row expansion of the determinant of a matrix with the repeated rows, hence the equation $B_{1s}=0$ is satisfied.
%
Note also that $A^k \neq 0$ for some $k$ since the rank of the matrix $P$ is $N-1$.

Now we will show that off-diagonal entries $B_{rt}=0,$ where $ 2\leq r<t\leq N$. In order to do so, we add a row corresponding to the non-diagonal entry  $B_{rt}$ to the coefficient matrix $P$ of the linear system \eqref{non-diagonal B1m} and we will show that the resulting matrix is singular. 
This will imply the existence of a non-trivial solution to the resulting system of $N$ equations.
Indeed, as the first $N-1$ equations have a unique solution given by \eqref{A_k in N-dim} up to proportionality, it also has to solve the last equation. 
 %
 %
 Thus we consider equations \eqref{non-diagonal B1m} together with
  \begin{equation}\label{system for Brt}
 B_{rt}=\sum_{k=1}^N A^k F_{krt}=0 
 \end{equation}
as a system of linear equations for functions $A^k$.
Its coefficient matrix is singular by Lemma \ref{sinqular matrix for rt-entry}.  
 %
This proves the statement.

\end{proof}
The following statement gives a further property of the matrix $B$.
\begin{proposition}\label{diagonal metric.N-dim}
Under the assumptions of Proposition \ref{non-diagonal metric.N-dim}
we have 
\begin{equation*}
    B_{11}=B_{pp} 
\end{equation*}
for all $1\le p \le N$.
\end{proposition}
\begin{proof}
Let us assume that $i_0 =1$, the general case can be dealt with similarly.
Let us first consider the case $p=2$.
Since the matrix $P$ 
 has rank $N-1$, this implies that there exists some $q \, \, (1\leq q \leq N) $ such that $F_{12q}\neq 0.$
Following the idea of the proof of Proposition \ref{non-diagonal metric.N-dim}, let us consider the following set of homogeneous equations:
 \begin{align}\label{system for diagonal}
 &B_{1m}=0,\quad 2\leq m \leq N, \nonumber \\ 
 &F_{12q}(B_{11}-B_{22})=0. 
 \end{align}
It is sufficient to show that the coefficient matrix $M$ corresponding to equations \eqref{system for diagonal} considered as linear equations for $A^k$ is singular.  
We have
 \begin{equation*}
 M = 
\begin{pmatrix}
F_{112} & F_{212} & \cdots  & F_{N12} \\
F_{113} & F_{213} & \cdots  & F_{N13} \\
\vdots  & \vdots  & \ddots & \vdots  \\
F_{11N} & F_{21N} & \cdots &  F_{N1N} \\
F_{12q}(F_{111}-F_{122}) & F_{12q}(F_{211}-F_{222}) & \cdots & F_{12q}( F_{N11}- F_{N22})
\end{pmatrix}.
\end{equation*}
Let $D=\det M$.
From the identity  \eqref{ij entry of FiFk} 
we have
\begin{equation}\label{iden1}
F_{12q}F_{111}=\sum_{m=1}^{N}F_{12m}F_{1qm}- \sum_{m=2}^N F_{11m}F_{2qm}.    
\end{equation}
Similarly, 
%
we have
\begin{equation}\label{iden2}
    F_{12q}F_{222}=\sum_{m=1}^{N}F_{12m}F_{2qm}- \sum_{m\neq 2} F_{1qm}F_{22m}.
\end{equation}
Let $R_i$ denote the $i^{th}$ row in the matrix $M$. 
We have 
\begin{align*}
    R_i & =(F_{11(i+1)},F_{21(i+1)},\dots, F_{N1(i+1)}), \quad 1\leq i \leq N-1,\nonumber\\
    R_N & = (r_{N1}, r_{N2}, \dots, r_{NN}), 
\end{align*}
where
\begin{align*}
    r_{N1} & =\sum_{m\neq 2} F_{1qm}F_{12m} -\sum_{m\neq 1} F_{11m}F_{2qm}, \nonumber\\
     r_{N2} & =\sum_{m\neq 2} F_{1qm}F_{22m} -\sum_{m\neq 1} F_{12m}F_{2qm}, \nonumber\\
   r_{Nk}&=  F_{12q}(F_{k11}-F_{k22}),\quad 3\leq k \leq N
\end{align*}
by applying formulas \eqref{iden1}, \eqref{iden2}.
Now let us perform the following row operation on the matrix $M$ and let $\widetilde{M}$ be the resulting matrix:
$$R_{N}\rightarrow \widetilde{R}_{N}= R_N -F_{11q}R_1+\sum_{i=2}^{N}F_{2qi}R_{i-1}.  $$
Let $\widetilde{r}_{Nk}$ be the $k^{th}$ element in the row $\widetilde{R}_{N}$ of the matrix $\widetilde{M}$. 
We have
\begin{align*}
  \widetilde{r}_{N1} & =\sum_{m\neq 1,2}F_{1qm} F_{12m}, \nonumber\\
    \widetilde{r}_{N2} & =\sum_{m\neq 1,2}F_{1qm} F_{22m}, \nonumber\\
   \widetilde{r}_{Nk} & = \sum_{m=1}^N F_{2qm}F_{1km} -F_{12q}F_{22k}-F_{11q}F_{12k},\quad 3\leq k \leq N.
\end{align*}
Let $S_{2m}=(F_{12m}, F_{22m},\dots, F_{N2m}),$ where $3 \leq m \leq N$.
By Lemma \ref{sinqular matrix for rt-entry} row $S_{2m}$ is a linear combination of the rows of the matrix $P$.
 %
 Therefore one can add the row $S_{2m}$
 to the last row of the matrix $\widetilde{M}$ without changing its determinant $D$.
 %
 %
 Let us add the rows $ -F_{1qm} S_{2m}, \,  m=3,\dots, N$, consecutively to the last row of $\widetilde{M}$.
 %
 The last row $(\widehat{r}_{N1},\widehat{r}_{N2},\dots , \widehat{r}_{NN})$ of the resulting matrix has the form
\begin{align*}\label{entries of the last row.v3}
   &\widehat{r}_{N1}= 0 , \quad    \widehat{r}_{N2} = 0, \nonumber\\
   & \widehat{r}_{Nk}  =  \sum_{m=1}^{N}F_{2qm}F_{1km}-\sum_{m=1}^NF_{1qm}F_{2km}=- [F_1,F_2]_{(q,k)},\quad 3\leq k \leq N
\end{align*}
by formula \eqref{ab entry}.
It follows from the commutativity equations that  $D=0$.
This proves that $B_{11}=B_{22}$. Similarly, one can prove that $B_{11}=B_{pp}$ for all $p$.
\end{proof}
As a corollary of Propositions \ref{non-diagonal metric.N-dim} and \ref{diagonal metric.N-dim} the following statement takes place.
\begin{theorem}\label{metric proportiona to idenity}
Under the assumptions of Proposition 
\ref{non-diagonal metric.N-dim}
the matrix $B$ given by formulas \eqref{general formula of the metric}, \eqref{A_k in N-dim}
is proportional to the identity matrix.
\end{theorem}
%
We also have the following result.
\begin{proposition}\label{non-deg linear combination of F}
Under the assumptions of Proposition 
\ref{non-diagonal metric.N-dim}
%
suppose also that there exists a non-degenerate linear combination $G=\eta^{k}F_k$ for some functions $\eta^k,  (1\leq k \leq N)$.
%
Then matrix $B$ given by formulas \eqref{general formula of the metric}, \eqref{A_k in N-dim} is a non-zero multiple of the identity matrix. 
\end{proposition}
\begin{proof}
From Theorem \ref{metric proportiona to idenity} we know that the matrix $B$ is proportional to the identity matrix. It remains to show that $B$ is not the zero matrix.
Let $B_{ij}= A^k F_{ijk}  =h\delta_{ij}$ for some function $h=h(x)$.
Assume that $h=0$. Then $A^k F_{ijk}=0$. Hence $\eta^{l} A^k F_{ljk}=0$, which means that the non-zero vector $(A^1 , \dots , A^N)$ belongs to the kernel of the form $G$ (cf. a similar argument in \cite{Misha&Veselov 2008}). Therefore  $G$ is degenerate, which contradicts the assumption. Hence $h\neq 0$ and the statement follows.  
\end{proof}

The following theorem is a corollary of Theorem \ref{metric proportiona to idenity} and Proposition \ref{non-deg linear combination of F}, and it explains that a function $F$ satisfying the commutativity equations also solves the WDVV equations.
%
\begin{theorem}\label{F solves WDVV and FF}
Under the assumptions of Proposition 
\ref{non-diagonal metric.N-dim} suppose that there exists a non-degenerate linear combination $G=\eta^{k}F_k$ for some functions $\eta^k$.
Then $F$ is a solution of WDVV equations \eqref{WDVV with metric B form} where the matrix $B$ is given by formulas \eqref{general formula of the metric}, \eqref{A_k in N-dim}.

\end{theorem}
%
 
\begin{remark}
Note that under the assumptions of Theorem \ref{F solves WDVV and FF}, function $F$ also satisfies the generalized WDVV equations 
\begin{equation*}
{F}_i {F}_{j}^{-1} {F}_k= {F}_k {F}_{j}^{-1} {F}_i, \quad i,j,k=1, \dots,N 
\end{equation*}
provided that matrices $F_j$ are non-degenerate. Indeed these equations follow from equations \eqref{WDVV with metric B form} 
by the result from \cite{Martini+Gragert 1999} (see also \cite{MMM.2000}). 
It also follows that $F$ satisfies the WDVV equations 
\begin{equation*}
{F}_i G^{-1} {F}_j= {F}_j G^{-1} {F}_i, \quad i,j=1, \dots,N
\end{equation*}
for any non-degenerate linear combination $G=a^i F_i$.
\end{remark}
\section{Existence of the identity field}\label{Existence of e}
In this Section we define a natural multiplication on the tangent plane $T_{x}V$ associated with a solution $F$ of the commutativity equations. 
We find the identity vector field of this multiplication and establish that it is proportional to the vector field $A^k \partial_{x^k}$,
where functions $A^k$ are given by formula \eqref{A_k in N-dim}.
%
%
Thus we will express the identity vector field in terms of $F$.

%
%
For any functions $ u=(u^1 , \dots , u^N),\, v=(v^1 , \dots , v^N)\colon  V \to V$, consider vector fields $\partial_u =u^i \partial_{x^i},\, \partial_v =v^i \partial_{x^i}\in \Gamma (TV)$. 
Let us define the following multiplication on the tangent space $T_{x}V$ for generic $x \in V \colon$
\begin{equation} \label{algebra.1}
\partial_{u} \ast \partial_{v}= u^{i} v^{j} F_{ijk}
\partial_{x^{k}}.
\end{equation}
Note that multiplication \eqref{algebra.1} defines a commutative associative algebra on $T_x V$ if $F$ satisfies commutativity equations \eqref{FF.ch4}.

Consider a vector field 
\begin{equation}\label{e the unity field}
e=e^{k}\partial_{x^{k}}, 
\end{equation}
where $e^{k}=e^{k}(x^{1},\dots,x^{N})$ 
are some functions.
Consider an $N\times N$ matrix $B=(B_{ij})_{i,j=1}^{N}$
%
given by
\begin{equation}\label{B matrix and e. formula}
B_{ij}=e(F_{ij})=e^{k} F_{ijk}, \quad i,j=1,\dots, N.
\end{equation}
%
%
\begin{proposition}\label{B matrix and e.2}
The following statements are equivalent:

(1) The matrix $B$ with entries given by  \eqref{B matrix and e. formula} is equal to the identity matrix,

(2) The vector field $e$ given by formula \eqref{e the unity field} is the identity vector field of the multiplication \eqref{algebra.1}. 
\end{proposition}
\begin{proof}
From relations \eqref{algebra.1}, \eqref{e the unity field} and \eqref{B matrix and e. formula} we have
\begin{align}\label{product w.r.t B}
 e \ast \partial_{v}= e^{i} v^{j}\partial_{x^{i}}\ast \partial_{x^{j}}
 = e^{i} v^{j} F_{ijk} \delta^{kl} \partial_{x^{l}}
 =B_{jk}  v^{j} \delta^{kl} \partial_{x^{l}}.
\end{align}
Let us firstly assume that $B_{jk}=\delta_{jk}$. Then relation \eqref{product w.r.t B} reduces to
\begin{align*}
 e \ast \partial_{v} 
 =  v^{j} \partial_{x^{j}} =\partial_{v}.
\end{align*}
That is statement (2) follows from (1).

Secondly, assume that $e$ 
is the identity vector field of the multiplication \eqref{algebra.1}. 
Then from relation \eqref{product w.r.t B} we have
\begin{align*}
 e \ast \partial_{v} 
 =B_{jk}  v^{j} \partial_{x^{k}}= \partial_{v} =   v^{j} \partial_{x^{j}}.
\end{align*}
This implies that $B_{jk}=\delta_{jk}$, that is statement (1) holds.
\end{proof}
%
 
Proposition \ref{B matrix and e.2} allows us to reformulate Theorem \ref{metric proportiona to idenity} as follows.
\begin{theorem}\label{generalised Th for FF with e}
Under the assumptions of Proposition 
\ref{non-diagonal metric.N-dim} there exists the identity vector field $e=e^k \partial_{x^k}$
for the multiplication \eqref{algebra.1}.
It has the form $e^k =h^{-1}A^k$, where $A^k$ is given by formula \eqref{A_k in N-dim} and
$h=A^k F_{kii}$ (for any $i=1,\dots, N$).
 \end{theorem}
%
 
%
%
Now we are going to generalize Theorem \ref{F solves WDVV and FF} to the case of an arbitrary constant metric $g$ in place of the standard metric $\delta_{ij}$. Thus we start with equations of the form
\begin{equation}\label{wdvv with const g}
    F_{i j \alpha} g^{\alpha \beta} F_{\beta k l}=
    F_{k j \alpha} g^{\alpha \beta} F_{\beta i l}, quad 1\le i, j, k, l \le N.
\end{equation}
We will show that matrix $(g_{\alpha \beta})_{\alpha, \beta =1}^N$ can be represented as a linear combination of the matrices $F_i$ under some non-degeneracy assumptions. 
\begin{theorem}\label{Theorem of unity field}
Let $F=F(x^1 ,\dots , x^N)$ be a function on $V$ which satisfies equations \eqref{wdvv with const g}
%
%
for some constant symmetric non-degenerate matrix $(g^{\alpha \beta})$.
%
Define new coordinates 
\begin{align}\label{changing variables}
 y^i =\widehat{C} _{j}^{i} x^j ,
\end{align}
where the matrix $\widehat{C}=(\widehat{C}^i_j)$
satisfies the relations $ \widehat{C}_{i}^{\alpha} \widehat{C}_{j}^{\beta} g^{ij}=\delta^{\alpha \beta}$, where $1 \le \alpha, \beta \le N$. 
%
%
Let $\widetilde{F}(y)=F(x)$. Suppose that there exists $i_0, 1\leq i_0 \leq N$, such that the matrix $(\widetilde{F}_{i_0 ij}(y))$ has rank $N-1$, where 
$\widetilde F_{i_0 i j} =\frac{\partial^3 \widetilde F}{\partial y^{i_0} y^i y^j}$ 
and $1\leq i, j\leq N, \, i \neq i_0$.   
Then there exists a unique vector field $e=e^{k}(x)\partial_{x^k}$ 
such that
\begin{equation*}
    e(F_{lm})=e^{k} F_{klm}=g_{lm},
\end{equation*}
where $(g_{lm})$ is the inverse matrix for $(g^{\alpha \beta})$.
\end{theorem}
%
%
%
%
\begin{proof}
%
 %
Let $C=\widehat C^{-1} = (C_k^i)$.  
Then $x^i = C_{j}^{i} y^j$.
We also have
\begin{equation}\label{basis change}
    \partial_{x^j}= \widehat{C}^{i}_{j} \partial_{y^i}, \quad 
     \partial_{y^j}= {C}^{i}_{j} \partial_{x^i}.
\end{equation}
%
%
%
%
%
From \eqref{basis change} we have the following relations:
\begin{equation*}
    F_j (x)=\widehat{C}^{i}_{j}\partial_{y^i}\widetilde{F}(y)=\widehat{C}^{i}_{j}\widetilde{F}_i(y).
\end{equation*}
%
%
%
%
Hence, 
\begin{align}\label{3rd derivitives after changing variables v1}
 F_{pjk}(x) = \widehat{C}^{m}_{p}\widehat{C}^{i}_{j}\widehat{C}^{l}_{k}\widetilde{F}_{mil}(y).
\end{align}
By multiplying relation \eqref{3rd derivitives after changing variables v1}
by $C_{a}^{p} C_{b}^{j} C_{c}^{k}$ we get 
\begin{align*}
 \widetilde{F}_{abc}(y)= C^{p}_{a} C^{j}_{b}  C^{k}_{c} F_{pjk}(x).
\end{align*}
%
%
By applying relation \eqref{3rd derivitives after changing variables v1} we rewrite equation  \eqref{wdvv with const g} as 
\begin{align}\label{wdvv with const g. v3}
    \widehat{C}_{i}^{p}  \widehat{C}_{j}^{q}  \widehat{C}_{\alpha}^{r} \widetilde{F}_{pqr}
    g^{\alpha \beta}
     \widehat{C}_{\beta}^{a}  \widehat{C}_{k}^{b}  \widehat{C}_{l}^{d} \widetilde{F}_{abd}
     =  \widehat{C}_{k}^{s}  \widehat{C}_{j}^{q}  \widehat{C}_{\alpha}^{r} \widetilde{F}_{sqr} 
     g^{\alpha \beta}
      \widehat{C}_{\beta}^{a}  \widehat{C}_{i}^{b}  \widehat{C}_{l}^{d} \widetilde{F}_{abd}.
\end{align}
It follows from the relation 
$\widehat{C}^{\alpha}_i \widehat{C}^{\beta}_j g^{ij}=\delta^{\alpha \beta}$
that
\begin{equation}\label{iden3}
    \widehat{C}^{l}_a \widehat{C}^{m}_b \delta_{lm}=g_{ab}.
\end{equation}
%
 Hence we reduce equation \eqref{wdvv with const g. v3} to
\begin{align}\label{wdvv with const g. v4}
    \widehat{C}_{i}^{p} 
    \widehat{C}_{k}^{b} 
   \widetilde{F}_{pqr}
    \widetilde{F}_{rbd}
     =  \widehat{C}_{k}^{s}  \widehat{C}_{i}^{b}  \widetilde{F}_{sqr} 
       \widetilde{F}_{rbd}.
\end{align}
By multiplying equation \eqref{wdvv with const g. v4} by $C_{n}^{i} C_{m}^{k}$ we get
\begin{align*}
   \widetilde{F}_{nqr}
    \widetilde{F}_{rmd}
     =   \widetilde{F}_{mqr} 
       \widetilde{F}_{rnd},
\end{align*}
that is $ \widetilde{F}_m$ and $ \widetilde{F}_n$ commute. 
Now since the rank of the matrix $ (\widetilde{F}_{i_0 ij})$ is  $N-1$, 
it follows by Theorem \ref{metric proportiona to idenity} that there exists a unique vector field $e=e^j (y) \partial_{y^j}$ 
such that 
\begin{equation*}
    e(\widetilde{F}_{\alpha \beta}(y))=e^j (y) \widetilde{F}_{j \alpha \beta}(y)=\delta_{\alpha \beta}.
\end{equation*}
%
%
From relation \eqref{3rd derivitives after changing variables v1} we have
$$C_{j}^{i}\partial_{x^i}F_{ab}(x)
=\widehat{C}_{a}^{l} \widehat{C}_{b}^{m}  \widetilde{F}_{jlm}(y).$$
This equation implies that
\begin{equation*}
   e(F_{ab}(x))
   =e^j C_{j}^{i}\partial_{x^i}F_{ab}(x)=
   e^j \widetilde{F}_{jlm}(y)\widehat{C}_{a}^{l} \widehat{C}_{b}^{lm}
    =\delta_{lm}\widehat{C}_{a}^{l} \widehat{C}_{b}^{m}=g_{ab}
\end{equation*}
by relation \eqref{iden3}. 
This proves the theorem. 
\end{proof}
\begin{remark}
\label{twodimrem1}
 We note that the maximality of the rank of the matrix $P$   is sufficient but not necessary for the existence of the identity field. Indeed, in the case of two-dimensional Frobenius manifold consider the function $F$ given by 
 \begin{equation}
 \label{prepdim2}
     F(t^1, t^2)=\frac{1}{2}(t^1)^2 t^2 +f(t^2)
 \end{equation} 
with some function $f(t^2)$.
 We have equation $F_1 G F_2= F_2 G F_1$, where 
the matrix entries $(F_i)_{kl}=\frac{\partial^3F}{\partial t^i \partial t^k \partial t^l}$, and 
\begin{equation}
G=G^{-1}=
 \begin{pmatrix}
 0 &1\\
 1 & 0\\
 \end{pmatrix}
 =F_1 .
\end{equation}
%
Now let 
\begin{align*}
C=
 \begin{pmatrix}
 -\frac{i}{2} &\frac{1}{2}\\
 i & 1\\
 \end{pmatrix}, \quad
 \widehat{C}=
 \begin{pmatrix}
i & -\frac{i}{2}\\
 1 &  \frac{1}{2}\\
 \end{pmatrix}
  \end{align*}
be the matrices of the change of variables such that 
\begin{align}\label{changing variables1,2}
 t^i ={C} _{j}^{i} x^j, \quad
 x^i =\widehat{C} _{j}^{i} t^j ,
\end{align}
where $(x^1, x^2)$ 
is a new coordinate system and the matrix $G=(g^{ij})$ satisfies the relation
$ \widehat{C}_{i}^{\alpha} \widehat{C}_{j}^{\beta} g^{ij}=\delta^{\alpha \beta}$. 
   Let $\widetilde{F}(x)=F(t)$.
   Then we have
   \begin{align*}
     \widetilde{F}_{112}&
 =\frac{\partial^3 \widetilde{F}}{\partial x^1 \partial x^1 \partial x^2}
 =\frac{1}{4}-f^{'''}(ix^1 +x^2),\\
 \widetilde{F}_{122}&
 =\frac{\partial^3 \widetilde{F}}{\partial x^1 \partial x^2 \partial x^2}
 =-\frac{i}{4}+if^{'''}(ix^1 +x^2).
   \end{align*}
 %
 Note that the matrix $( \widetilde{F}_{112}, \widetilde{F}_{122})$ has rank zero if  $f(t^2)=\frac{1}{24}(t^2)^3=\frac{1}{24}(ix^1+x^2)^3$.
 Nonetheless $e=\partial_{t^1}=i\partial_{x^1}+ \partial_{x^2}$ is the identity field.
 \end{remark}
 
 \begin{remark}
 The maximality of the rank of the matrix $P$ condition may be satisfied in the case of a family of non-semisimple algebras. An example is given by prepotential \eqref{prepdim2} with $f=0$, as it follows from considerations in Remark \ref{twodimrem1}.
 \end{remark}

%
%

\section{Identity field for non-simply laced root systems and their projections}\label{Section.Uniform formula}
%

%
%
%
%
 %

We are going to relate the identity field for a solution of commutativity equations and a restriction of such a solution. Firstly, we have the following statement.


%
\begin{lemma}\label{e under orthogonal transformation}
Let $F(x)$ be a function and $e=e^{k}\partial_{x^{k}}$ be a vector field such that $e(F_{ij})=\delta_{ij}$, where $1\leq i,j \leq N$.
Let $\widetilde{x}^k = C^k_i x^i$ for a matrix $C=(C^k_i)\in O(N,\mathbb C)$.
%
Let $\widetilde{F}(\widetilde{x})=F(x)$. 
%
Then ${e}(\widetilde{F}_{\mu \nu})=\delta_{\mu\nu}$, where $1\leq \mu ,\nu \leq N$.
%
%
\end{lemma}
\begin{proof}
We have
\begin{equation}\label{e(F)=delta}
e^k F_{ijk}=\delta_{ij}.   
\end{equation}
By relation \eqref{changing variables.V2}
we have 
$\partial_{x^i}=C_i^k \partial_{\widetilde{x}^k}$.
Hence we have
\begin{align}\label{change F.2}
F_{ijk}={C}_{k}^{\widetilde{k}}{C}_{j}^{\widetilde{j}}{C}_{i}^{\widetilde{i}}\widetilde{F}_{\widetilde{i}\widetilde{j}\widetilde{k}}.   
\end{align}
%
Then by formula \eqref{change F.2} relation \eqref{e(F)=delta} can be written as
\begin{align}\label{change F.3}
 {e}^k 
 {C}_{k}^{\widetilde{k}} {C}_{j}^{\widetilde{j}} {C}_{i}^{\widetilde{i}}
 \widetilde{F}_{\widetilde{i}\widetilde{j}\widetilde{k}}
=\delta_{ij}.   
\end{align}
Let $\widehat{C}=C^{-1}=(\widehat{C}^i_j)$. 
Multiply equality \eqref{change F.3} by $\widehat{C}_{\mu}^{i} \widehat{C}_{\nu}^{j}$. We get
\begin{align}\label{change F.4}
 {e}^k 
 {C}_{k}^{\widetilde{k}} 
 \widetilde{F}_{\mu \nu\widetilde{k}}
 =\widehat{C}_{\mu}^{i} \widehat{C}_{\nu}^{i} = \delta_{\mu \nu}
\end{align}
since $\widehat{C}\in O(N, \mathbb C)$. 
Hence equality \eqref{change F.4} becomes
\begin{align}\label{change F.5}
 {e}^k 
 {C}_{k}^{\widetilde{k}} 
 \widetilde{F}_{\mu \nu\widetilde{k}}
=\delta_{\mu \nu}.
\end{align}
Note that
 ${e}={e}^k C_{k}^{\widetilde{k}}\partial_{x^{\widetilde{k}}}$.
%
We have by relation \eqref{change F.5} that ${e}(\widetilde{F}_{\mu \nu})=\delta_{\mu \nu}$ as required.
\end{proof}

Let 
$ \mathcal{B}=\mathcal{A} \cap W$
be a subsystem of $\mathcal{A}$ for some  $n$-dimensional linear subspace $ W=\langle \mathcal{B} \rangle \subseteq V.$ 
Let
$$W_{\mathcal{B}} \coloneqq \{x \in V \colon (\beta,x)=0 \quad \forall \beta \in \mathcal{B}\}.   $$
%
%
%
%
 %
%
%
%
Let $f_1 ,\dots, f_n$ be an orthonormal basis of the space $W_{\mathcal{B}}$, 
and let ${\xi^1,\dots,\xi^n}$ be the corresponding orthonormal coordinates in $W_{\mathcal{B}}$.
Let us extend the orthonormal basis in $W_{\mathcal{B}}$ to an orthonormal basis $f_1 ,\dots, f_n, f_{n+1},\dots , f_N$ in $V$ and let ${\xi^1,\dots,\xi^n, \xi^{n+1},\dots , \xi^N}$ be the corresponding orthonormal coordinates in $V$. 
%
%
The following statement takes place.
\begin{proposition}\label{identity of rest F.1}
Let a function $F$ be given by formula \eqref{trig.solution. general form 2}.
Let $e=e(z), \, z\in V$, be a vector field such that $e(F_{ij})=\delta_{ij}$ for all $i,j=1,\dots ,N$, where
$F_{ij}=\frac{\partial^{2} F}{\partial x^i \partial x^j}$. 
Let $\widehat{F}(\xi^1 ,\dots ,\xi^N)=F(x^1,\dots, x^N)$. 
Suppose that $e(\xi)\in T_{{\xi}}W_{\mathcal{B}}$ for a generic ${\xi}\in W_{\mathcal{B}}$.
Let $\widetilde{e}={e}|_{W_{\mathcal{B}}}\in \Gamma (T_{\ast}W_{\mathcal{B}})$. Then
$\widetilde{e}((F_{\mathcal{B}})_{ij})=\delta_{ij},$
where
$(F_{\mathcal{B}})_{ij}=\frac{\partial^{2} F_{\mathcal{B}}}{\partial \xi^i \partial \xi^j}$
and
function $F_{\mathcal{B}}$ is given by formula \eqref{F_B solution for FF}.
\end{proposition}
\begin{proof}
Let $C=(C^k_i)\in O(N, \mathbb{C})$ be such that $\xi^k = C_{i}^{k}x^i$.
By Proposition \ref{e under orthogonal transformation} we have ${e}(\widehat{F}_{ij})=\delta_{ij}$,  where $\widehat{F}_{ij}=\frac{\partial^{2} \widehat{F}}{\partial \xi^i \partial \xi^j}$, and  $i,j=1,\dots, N$.
Hence $\widetilde{e}(\widehat{F}_{ij}|_{W_\mathcal{B}})=\delta_{ij},\, 1\le i,j\le n$, which implies 
the statement since 
$\widehat{F}|_{W_{\mathcal{B}}}= F_{\mathcal{B}}$
and
$\widehat{F}_{ij}|_{W_{\mathcal{B}}}= (F_{\mathcal{B}})_{ij}$.
\end{proof}

 In the next proposition we give a formula for the identity field for the multiplication \eqref{algebra.1} corresponding to the root system $F_4$, see \cite{Maali.2022} for a proof.
\begin{proposition}\label{e for F4.new formula}
The matrix $B=h^{-1}\sum_{k=1}^4 B^k F_k $ is the identity matrix in dimension four in the following cases: 
\begin{itemize}
    \item $F$ has the form  \eqref{trig.solution. general form 2}
corresponding to $\mathcal{A}=F_4^{+}$ with the condition $r=-2q$, $q\ne 0$, where
     \begin{align*}
   &B^k =\sin{x^k}\Big(\cos{x^k}(-1+\sum_{i\neq k}\cos{2x^i})-2\prod_{i\neq k}\cos{x^i}\Big), \quad k=1,2,3,4,
     \nonumber\\
    &h(x)= 6 q + \frac12 \sum_{\alpha \in F_{4}^{+}} c_{\alpha}\cos{(2\alpha ,x)}.
 \end{align*}
   \item $F$ has the form  \eqref{trig.solution. general form 2}
corresponding to $\mathcal{A}=F_4^{+}$ with the condition $r=-4q$, $q\ne 0$,  where
     \begin{align*}
    &B^k =\sin{x^k}\Big(\cos{x^k}+2\prod_{i\neq k}\cos{x^i}\Big), \quad k=1,2,3,4,
     \nonumber\\
    &h(x)= -q \Big( 6 + \sum_{i=1}^{4}\cos{2x^i} + 8 \prod_{i=1}^{4}\cos{x^i} \Big).
 \end{align*}
\end{itemize}
\end{proposition}

Solutions of the WDVV equations corresponding to the root system $BC_n$ and its deformation $BC_n (q,r,s;\underline{m})$ were found  in \cite{MGM 2020}.
In the case of the root system $F_4$ and its projections we get new solutions of the WDVV equations.
 \begin{theorem}\label{WDVV and comm for F4 and restrictions}
 Function \eqref{trig.solution. general form 2}
corresponding to $\mathcal{A}=F_4^{+}$ or any of its $3$-dimensional projections $(F_4, A_1)_1, \, (F_4, A_1)_2$ satisfies WDVV equations \eqref{WDVV with metric B form} if $r=-2q$ or $r=-4q$, $q\ne 0$. 
 \end{theorem}
\begin{proof}
It was proven in \cite{George+Misha 2019} that
function \eqref{trig.solution. general form 2}
for the collection $\mathcal{A}=F_4^{+}$ satisfies commutativity equations \eqref{FF.ch4} if $r=-2q$ or $r=-4q$. 
For $\mathcal{A}=F_4^{+}$ the statement follows by Proposition \ref{e for F4.new formula}.
It is easy to see that for the  three-dimensional restrictions $\mathcal{A}=(F_4, A_1)_{1,2}$ 
the assumptions of Proposition \ref{identity of rest F.1} hold. The statement follows. 
\end{proof} 
 %

Now we give the identity vector field for all the non-simply laced root systems as well as their projections. In the case of root system $F_4$ it can be checked that the components of the identity field given by the next theorem are equal to $h^{-1}B^k$ given by Proposition \ref{e for F4.new formula} (see \cite{Maali.2022}). 

\begin{theorem}\label{universal e}
Let function $F$ be given by \eqref{trig.solution. general form 2}. 
Consider a vector field $e$ given by 
\begin{equation}\label{e for general root sys}
    e=c_0 H^{-1}\sum_{\alpha\in \mathcal{A}}\bar{c}_{\alpha}\sin(2(\alpha,x))\partial_{\alpha}
\end{equation}
for some $c_0,\, \bar{c}_{\alpha}\in \mathbb{C}$ and
 $$H=H_0 + \sum_{\alpha\in \mathcal{A}}\bar{c}_{\alpha}\sin^2(\alpha,x)$$
 for some $H_0\in \mathbb{C}$.
Then $e(F_{ij})=\delta_{ij}$ if
\begin{enumerate}
    \item\label{it1}
    $\mathcal{A}=F_{4}^{+}$ given by formula \eqref{F4.half positive} 
    or  $\mathcal{A}$ is one of the projections
   $(F_{4},A_1)_1,\, (F_{4},A_1)_2$,\, 
   $(F_{4},A_2)_1,\, (F_{4},A_2)_2,\, (F_{4},B_2),\, (F_{4},A_1^2) $, 
and
    $$r=-2q\neq 0,  \quad c_0 = -\frac{1}{4q}, \quad H_0=0 , \quad \bar{c}_{\alpha}={c}_{\alpha}\, \forall \alpha\in \mathcal{A},$$
 \item
 $\mathcal{A}$ is the same as in \eqref{it1} 
and
$$r=-4q\neq 0, \quad c_0 =\frac{1}{4q},\quad   H_0 =36q,\quad \bar{c}=c_{\alpha}|_{q=0}\, \forall \alpha \in \mathcal{A},$$
\item\label{it3}
     $\mathcal{A}=G_{2}^{+}$ 
   and
     $$p=-3q\neq 0,\quad c_0 =-\frac{1}{9q},\quad H_0=0 , \quad
     \bar{c}_{\alpha}={c}_{\alpha}\, \forall \alpha\in G_{2}^{+},$$
     where $q$ is the multiplicity of the long roots $\sqrt{3}e_1, \frac{1}{2}(\sqrt{3}e_1\pm 3e_2)$ and $p$ is the multiplicity of the short roots $e_2, \frac{1}{2}(\sqrt{3}e_1\pm e_2)$,
\item
 $\mathcal{A}$ is the same as in $\eqref{it3}$,
and 
$$p=-9q\neq 0,\quad c_0 =\frac{1}{9q},\quad H_0 =27 q , \quad 
\bar{c}=c_{\alpha}|_{q=0}\, \forall \alpha \in G_{2}^{+},$$
 \item
$\mathcal{A}=BC_n (q,r,s;\underline{m}), q\neq 0, n\geq 2,$ 
and
$$r=-8s-2q(\sum_{i=1}^{n}m_i -2),\quad c_0 =-\frac{1}{4q},\quad H_0 =\frac{r(2s-q)}{q},\quad
\bar{c}_{\alpha}={c}_{\alpha}|_{q=s=0}\, \forall \alpha\in BC_n (q,r,s;\underline{m})
.$$
\item
$\mathcal{A}=BC_1^+$ with $c_{\pm e_1}=r, c_{\pm 2e_1}=s$ and 
$$c_0 =-\frac{1}{2(r+8s)},\quad H_0 = - \frac{r(r+4s)}{r+8s},\quad \bar{c}_{e_1}=r,\, \bar{c}_{2e_1}=0.$$
 \end{enumerate}
\end{theorem}
 
Theorem \ref{universal e} follows from the identity 
\begin{equation}\label{new identity for e}
\sum_{\alpha,\beta \in \mathcal{A}}\bar{c}_{\alpha}c_{\beta} (\alpha,\beta) (\beta , u)(\beta , v)\sin(2\alpha,x)\cot(\beta,x)= c_0^{-1} H (u,v)
\end{equation}
for any $u,v \in V$ for each case specified in Theorem \ref{universal e}.
By Proposition \ref{identity of rest F.1}
it is sufficient to establish identity \eqref{new identity for e} for the case when $\mathcal{A}$ is a (non-simply laced) root system. 
Indeed it is easy to see that the vector field $e$ given by \eqref{e for general root sys} for $\mathcal{A}=F_4, \, BC_N$ satisfies the condition $e|_{W}\in \Gamma (T_{\ast}W)$ for any intersection of mirrors $W$.
It is also clear that the restricted vector field $e|_{W}$ has the form \eqref{e for general root sys} for the corresponding projections of the root system $\mathcal{A}$.
A case by case proof of the identity \eqref{new identity for e} for $\mathcal{A}=F_4$ and $\mathcal{A}=G_2$ is contained in \cite{Maali.2022}; see \cite{MGM 2020} for $\mathcal{A}=BC_N$. 
We are not aware of a uniform proof of Theorem \ref{universal e} or the identity \eqref{new identity for e}.  

\section*{Acknowledgments}
We thank I. Strachan, O. Mokhov, A.P. Veselov, E. Ferapontov 
and G. Cotti for useful discussions and comments. The work of M.A was funded by Department of Mathematics, College of Science and Humanities, Imam Abdulrahman Bin Faisal University, P.O. Box 12020 Jubail Industrial City 31961, Saudi Arabia.

%

\end{document}